%% file: manip_multiwinner_long_arxiv.tex
\documentclass[11pt]{article}
\usepackage{times}

\usepackage{amsmath}
\usepackage{amsthm}
\usepackage{amssymb}

\usepackage{multirow}
\usepackage{booktabs}
\RequirePackage[authoryear,sort,semicolon]{natbib}
\makeatletter
\renewcommand*{\NAT@spacechar}{ }
\makeatother
\usepackage[pdfdisplaydoctitle,colorlinks,citecolor=green,linkcolor=green,
pagebackref=true]{hyperref}
\usepackage[hyphenbreaks]{breakurl}
\let\oldhref\href
\renewcommand{\href}[2]{\oldhref{#1}{\hbox{#2}}}

\usepackage{etoolbox}
\usepackage{todonotes}

\usepackage{bm}

\newcommand{\mypar}[1]{\medskip\textbf{#1}\medspace}

\usepackage{subcaption}
\usepackage[algoruled, linesnumbered]{algorithm2e}

\usepackage{tikz}
\usetikzlibrary{arrows, decorations.pathmorphing, backgrounds, positioning, fit,
shapes.geometric, patterns, shapes.misc, calc}

\usepackage{framed}
\usepackage{tabularx}

\usepackage{authblk}

\newlength{\RoundedBoxWidth}
\newsavebox{\GrayRoundedBox}
\newenvironment{GrayBox}[1]%
   {\setlength{\RoundedBoxWidth}{.93\columnwidth}
    \def\boxheading{#1}
    \begin{lrbox}{\GrayRoundedBox}
       \begin{minipage}{\RoundedBoxWidth}}%
   {   \end{minipage}
    \end{lrbox}
    \begin{center}
    \begin{tikzpicture}%
     \node(Text)[draw=black!20,fill=white,rounded corners,inner sep=2ex,text width=\RoundedBoxWidth]
           {\usebox{\GrayRoundedBox}};
      \coordinate(x) at (current bounding box.north west);
      \node [draw=white,rectangle,inner sep=3pt,anchor=north west,fill=white]
      at ($(x)+(6pt,.75em)$) {\boxheading};
    \end{tikzpicture}
    \end{center}}

\newenvironment{defproblemx}[1]{\noindent\ignorespaces%
                                \FrameSep=0pt%
                                \parindent=0pt%
                \vspace*{-1em}%
                \begin{GrayBox}{#1}%
                \begin{tabular*}{\RoundedBoxWidth}{@{} >{\itshape}
		 p{.15\RoundedBoxWidth} @{} p{.85\RoundedBoxWidth} @{}}%
            }{
                \end{tabular*}%
                \end{GrayBox}%
                \ignorespacesafterend
            }

\newcommand{\np}{\ensuremath{\mathrm{NP}}}

\newcommand{\fpt}{\ensuremath{\mathrm{FPT}}}
\newcommand{\xp}{\ensuremath{\mathrm{XP}}}
\newcommand{\wone}{\ensuremath{\mathrm{W[1]}}}
\newcommand{\wtwo}{\ensuremath{\mathrm{W[2]}}}

\newcommand{\udiff}{\ensuremath{u_{\diff}}}

\DeclareMathOperator{\score}{score}

\DeclareMathOperator{\lexpref}{>_\mathcal{\mathcal{F}}}

\newcommand{\strpref}[1]{\ensuremath{\operatorname{>}_{\mathcal{S}}^{#1}}}

\DeclareMathOperator{\lex}{lex}
\DeclareMathOperator{\opt}{opt}
\DeclareMathOperator{\pess}{pess}
\DeclareMathOperator{\util}{util}
\DeclareMathOperator{\egal}{egal}
\DeclareMathOperator{\candegal}{candegal}
\DeclareMathOperator{\eval}{eval}
\DeclareMathOperator{\bhav}{bhav}

\DeclareMathOperator{\secured}{C^*}
\DeclareMathOperator{\mful}{Mful}
\DeclareMathOperator{\fbid}{Fbid}
\DeclareMathOperator{\used}{\overline{used}}
\DeclareMathOperator{\fused}{\overline{fullyused}}
\DeclareMathOperator{\diff}{diff}

\newcommand{\naturals}{\ensuremath{{\mathbb{N}}}}
\newcommand{\calR}{\ensuremath{{{\mathcal{R}}}}\xspace}
\newcommand{\calF}{\ensuremath{{{\mathcal{F}}}}\xspace}

\newcommand{\Cconf}{\ensuremath{C^+}}
\newcommand{\Cpend}{\ensuremath{P}}
\newcommand{\Crej}{\ensuremath{C^-}}

\newcommand{\CM}[3]{${#1}$-${#2}$-${#3}$-\textsc{CM}\xspace}
\newcommand{\CMlong}[3]{${#1}$-${#2}$-${#3}$-\textsc{Co\-a\-li\-tio\-nal
Ma\-nipu\-la\-tion}\xspace}
\newcommand{\CCM}[3]{${#1}$-${#2}$-${#3}$-\textsc{CM} with consistent manipulators\xspace}
\newcommand{\CCMlong}[3]{${#1}$-${#2}$-${#3}$-\textsc{Coalitional Manipulation}
with consistent manipulators\xspace}

\newcommand{\SNTV}{\text{SNTV}}
\newcommand{\Bloc}{\text{Bloc}}
\newcommand{\lBloc}[1]{\ensuremath{#1}\text{-\Bloc{}}}
\newcommand{\ellBloc}{\lBloc{\ell}}

\newcommand{\Flex}{\ensuremath{\calF_{\lex}}}

\newcommand{\Fopt}{\ensuremath{\calF_{\opt}}}
\newcommand{\Fpess}{\ensuremath{\calF_{\pess}}}

\newcommand{\TIE}[2]{\textsc{$\calF^{#1}_{#2}$-TB}\xspace}
\newcommand{\TIElong}[2]{\textsc{$\calF^{#1}_{#2}$-Tie-Break\-ing}\xspace}

\newcommand{\elitegroup}{excellence-group}
\newcommand{\elgroup}{egroup}
\newcommand{\namedelitegroup}[1]{\ensuremath{#1}-\elitegroup}
\newcommand{\namedelgroup}[1]{\ensuremath{#1}-\elgroup}
\newcommand{\kelgroup}{\namedelgroup{k}}
\newcommand{\kelitegroup}{\namedelitegroup{k}}
\newcommand{\obligatory}{\ensuremath{\operatorname{obl}}}



\newcommand{\probDefLong}[5][]{%
 \begin{defproblemx}{#2 (#3)}
  \ifx&#1&%
  \else%
  &#1\\
  \fi\\%
  \textbf{Input:} & #4\\
  \textbf{Question:} & #5
 \end{defproblemx}%
}

\newtheorem{theorem}{Theorem}
\newtheorem{proposition}{Proposition}

\newtheorem{corollary}{Corollary}

\newtheorem{observation}{Observation}

\theoremstyle{definition}
\newtheorem{definition}{Definition}

\newtheorem{example}{Example}

\newtheorem{claim}{Claim}

\newcommand{\candPalong}{Uematsu: Final Fantasy}
\newcommand{\candPblong}{Badelt: Pirates of the Caribbean}

\newcommand{\candMalong}{Mozart: Clarinet Concerto}
\newcommand{\candMblong}{Mozart: Jeunehomme Piano Concerto}
\newcommand{\candBalong}{Beethoven: Piano Concerto No.\,5}
\newcommand{\candBblong}{Beethoven: Symphony No.\,6}
\newcommand{\candPa}{\ensuremath{o_1}\xspace}
\newcommand{\candPb}{\ensuremath{o_2}\xspace}

\newcommand{\candMa}{\ensuremath{m_1}\xspace}
\newcommand{\candMb}{\ensuremath{m_2}\xspace}
\newcommand{\candBa}{\ensuremath{b_1}\xspace}
\newcommand{\candBb}{\ensuremath{b_2}\xspace}
\newcommand{\vvoteYa}{\ensuremath{v_{\text{y}_1}}\xspace}
\newcommand{\vvoteYb}{\ensuremath{v_{\text{y}_2}}\xspace}
\newcommand{\vvoteYc}{\ensuremath{v_{\text{y}_3}}\xspace}
\newcommand{\voteMa}{\ensuremath{v_{\text{m}_1}}\xspace}
\newcommand{\voteMb}{\ensuremath{v_{\text{m}_2}}\xspace}
\newcommand{\voteBa}{\ensuremath{v_{\text{b}_1}}\xspace}
\newcommand{\voteBb}{\ensuremath{v_{\text{b}_2}}\xspace}
\newcommand{\exampleCandidates}{\ensuremath{C=\{\candBa,\allowbreak
 \candBb,\allowbreak \candMa,\allowbreak \candMb,\allowbreak \candPa,\allowbreak
 \candPb\}}}

\include{commons}

\newcommand{\manipnumber}{\ensuremath{r}}

\providecommand{\keywords}[1]{\par\textbf{\textit{Keywords:}} #1}

\begin{document}
\title{On Coalitional Manipulation for Multiwinner Elections:
  Shortlisting\footnote{A preliminary version of this article appeared in
   the \emph{Proceedings of the Twenty-Sixth International Joint Conference on Artificial
   Intelligence (IJCAI '17)}~\citep{BKN17}.}}
\author{Robert Bredereck}
\author{Andrzej Kaczmarczyk}
\author{Rolf Niedermeier}
\affil{%
 TU Berlin,
 Faculty~IV,
 Algorithmics and Computational Complexity,\\
 Ernst-Reuter-Platz 7, D-10587 Berlin, Germany%
}
\maketitle
\begin{abstract}
Shortlisting of candidates---selecting a group of~``best'' candidates---is 
a special case 
of multiwinner elections.
We provide the first in-depth study of the computational complexity of strategic
voting for shortlisting based on the perhaps most basic voting rule in this
scenario, \ellBloc{} (every voter approves $\ell$~candidates). In particular, we
investigate the influence of several tie-breaking mechanisms (e.g.,\ pessimistic
versus optimistic) and group evaluation functions (e.g.,\ egalitarian versus
utilitarian). Among other things, conclude that in an egalitarian setting
strategic voting may indeed be computationally intractable regardless of the
tie-breaking rule. Altogether, we~provide a~fairly comprehensive picture of the
computational complexity landscape of this scenario.
\end{abstract}

\section{Introduction}
A university wants to select the two favorite pieces in classical style to
be played during the next graduation ceremony.
The students were asked to submit their favorite pieces. Then a jury
consisting of seven members (three juniors and four seniors) from the university
staff selects from the six most frequently submitted pieces as follows:
Each jury member approves two pieces and the two winners are those
obtaining most of the approvals.
The six options provided by the students are
``\candBalong{}~(\candBa)'',
``\candBblong{}~(\candBb)'',
``\candMalong{}~(\candMa)'',
``\candMblong{}~(\candMb)'',
``\candPalong{}~(\candPa)'', and
``\candPblong{}~(\candPb).''
The three junior jury members are excited about recent audio-visual presentation
arts (both interactive and passive) and approve \candPa and \candPb.
Two of the senior jury members are Mozart enthusiasts, and the other two senior jury
members are Beethoven enthusiasts. Hence, when voting truthfully, two of them
would approve the two Mozart pieces and the other two would approve the
two Beethoven pieces.
The winners of the selection process would be \candPa and \candPb, both
receiving three approvals whereas every other piece receives only two
approvals.

The senior jury members meet every Friday evening and discuss important academic
issues including the graduation ceremony music selection processes,
why ``movie background noise'' recently counts as classical
music,\footnote{\url{http://www.classicfm.com/radio/hall-of-fame/}} and the
influence of video games on the ability of making important decisions.
During such a meeting they agreed that a graduation ceremony should always be
accompanied by pieces of traditional, first-class composers. Thus, finally all
four senior jury members decide to approve~\candBa and~\candMa so these two
pieces are played during the graduation ceremony.

Already this toy example above (which will be the basis of our running example
throughout the paper) illustrates important aspects of strategic voting in
multiwinner elections.
In case of coalitional manipulation for single-winner elections (where a
coalition of voters casts untruthful votes in order to influence the outcome of an
election; a topic which has been intensively studied in the
literature~\citep{Rothe15,handbookCOMSOC})
one can always assume that a coalition of manipulators agrees on trying to make
a distinguished alternative win the election.
In case of \emph{multiwinner elections}, however, already determining concrete
possible goals of a coalition seems to be a non-trivial task:
There may be exponentially many different outcomes which can be reached through
strategic votes of the coalition members and each member could have its
individual evaluation of these outcomes.

Multiwinner voting rules come up very naturally whenever one has to select from
a large set of candidates a smaller set of ``the best'' candidates.
Surprisingly, although at least as practically relevant as single-winner voting
rules, the multiwinner literature is much less developed than the single-winner
literature. In recent years (see a survey of \citet{COMSOCtrends17multi}),
however, research into multiwinner voting rules, their properties, and
algorithmic complexity grew significantly~\citep{ABCEFW17, BSU13, AEFLS17,
 AGGMMW15, BGLMR13, BC08, BC10, EFSS17, FSST16, FST17, MPRZ08, OZE13, Sko16,
SFS15, LS18}. When selecting a group of winning candidates, various criteria can be
interesting; namely, proportional representation, diversity, or excellence
(see~\citet{EFSS17}). We focus on the last scenario. Here the goal is to select
the best (say highest-scoring) group of~candidates obtaining a so-called
\emph{short list}.

Shortlisting comes very naturally in the context of selection committees; for
instance, for human resources departments that need to select, for a fixed
number of positions, the best qualified applicants. A standard way of candidate
selection in the context of shortlisting is to use scoring-based voting
rules.
We focus on the two most natural ones: \SNTV{} (single non-transferable
vote---each voter gives one point to one candidate) and \ellBloc{} (each voter
gives one point to each of $\ell$ different candidates, so \SNTV{} is the same
as \lBloc{1}).\footnote{Although, in general, \ellBloc{} does not satisfy
\emph{committee monotonicity} which is considered as a necessary condition for
shortlisting~\citep{COMSOCtrends17multi}, this rule seems quite frequent in
practice---for example The Board of Research Excellence in Poland was elected
using a variant of~\ellBloc{}~\citep{RDNVote}.} Obviously, for such voting rules
it is trivial to determine the score of each individual candidate. The main goal
of our work is to model and understand coalitional manipulation in a
computational sense---that is, to introduce a formal description of how a group
of manipulators can influence the election outcome by casting strategic votes
and whether it is possible to find an effective strategy for the manipulators to
change the outcome in some desired way. In this fashion, we complement
well-known work: manipulation for single-winner rules initiated
by~\citet{BTT89}, coalitional manipulation for single-winner rules initiated
by~\citet{CSL07}, and (non-coalitional) manipulation for multiwinner rules
initiated by~\citet{MPRZ08}. In coalitional manipulation scenarios, given full
knowledge about other voters' preferences, one has a set of manipulative voters
who want to influence the election outcome in a favorable way by casting their
votes strategically.

To come up with a useful framework for coalitional manipulation for multiwinner
elections, we first have to identify the exact mathematical model and questions
to be asked. A couple of straightforward extensions of coalitional manipulation
for single-winner elections or (non-coalitional) manipulation for multiwinner
elections do not fit. Extending the single-winner variant directly, one would
probably assume that the coalition agrees on making a distinguished candidate
part of the winners or that the coalition agrees on making a distinguished
candidate group part of the winners. The former is unrealistic because in
multiwinner settings one typically cares about more than just one
candidate---especially in shortlisting it is natural that one wants rather some
group of ``similarly good'' candidates to be winning instead of only one
representative of such a group. The latter, that is, agreeing on a distinguished
candidate group to be part of the winners is also problematic since there may be
exponentially many ``equally good'' candidate groups for the coalition. Notably,
this was not a problem in the single-winner case; there, one can test for a
successful manipulation towards each possible candidate avoiding an exponential
increase of the running time (compared to the running time of such a test for a
single candidate).

We address the aforementioned issue of modeling coalitional manipulation for
multiwinner election by extending a single-manipulator model for multiwinner
rules of~\citet{MPRZ08}. In their work, the manipulator specifies the utility of
each candidate and the utility for a candidate group is obtained by adding up
the utilities of each group member. We build on their idea and let each
manipulator report the utility of each candidate. However, aggregating
utilities for a coalition of manipulators (in other words, computing a
collective utility of manipulators) becomes nontrivial---especially for a
coalition of manipulators who have diversified utilities for single candidates
but still have strong incentives to work together (e.g.,\ as we have seen in our
introductory example).

\mypar{Our Contributions.}
We devise a formal description of coalitional manipulation in multiwinner
elections arriving at a new, nontrivial model capturing two~types of
manipulators' attitudes and a few natural ways of utility aggregation. To this
end, in our model, we distinguish between optimistic and pessimistic
manipulators and we formalize aggregation of utilities in a utilitarian and an
egalitarian way.

Using our model, we analyzed the computational complexity of finding a
successful manipulation for a coalition of voters, assuming elections under
rules from the family of~\ellBloc{} voting rules. We show that, even for these
fairly simple rules, the computational complexity of coalitional manipulation is
diverse. In particular, we observed that finding a manipulation maximizing the
utility of a worst-off manipulator (egalitarian aggregation) is~\np-hard
(regardless of the manipulators' attitude). This result stands in sharp contrast to
the polynomial-time algorithms that we give for finding a manipulation
maximizing the sum of manipulators' utilities (utilitarian aggregation).
Additionally, we show how to circumvent the cumbersome~\np-hardness for the
egalitarian aggregation providing an (FPT) algorithm that is efficient for
scenarios with few manipulators and few different values of utility that
manipulators assign to agents. We survey all our computational complexity
results in~\autoref{tab:results} (\autoref{sec:conclusion}).

\mypar{Related Work.}
To the best of our knowledge, there is no previous work on \emph{coalitional}
manipulation in the context of multiwinner elections. We refer to recent
textbooks for an overview of the huge literature on single-winner (coalitional)
manipulation \citep{Rothe15,handbookCOMSOC}. Most relevant to our work,
\citet{Lin11} showed that coalitional manipulation in single-winner elections
under $\ell$-Approval is solvable in linear time by a greedy algorithm.
\citet{MPRZ08} introduced (non-coalitional) manipulation for multiwinner
elections. While pinpointing manipulation for several voting rules as NP-hard,
they showed that manipulation remains polynomial-time solvable for
\Bloc{} (which can be interpreted as a multiwinner equivalent of $1$-Approval).
\citet{OZE13} extended the latter result for different tie-breaking strategies
and identified further tractable special cases of multiwinner scoring rules but
conjectured manipulation to be hard in general for (other) scoring rules.
Summarizing, \Bloc{} is simple but comparably well-studied, and hence we
selected it as a showcase for our study of the presumably computationally harder
\emph{coalitional} manipulation.

\mypar{Organization.}
\autoref{sec:Prelim} introduces basic notation and formal concepts.
In \autoref{sec:model}, we develop our model for coalitional manipulation in
multiwinner elections. Its variants respect different ways of evaluating
candidate groups (utilitarian vs.\ egalitarian) and two kinds of
manipulators behavior (optimistic vs.\ pessimistic).
In \autoref{sec:compl:tie-breaking}, we present algorithms and complexity
results for computing the output of several tie-breaking rules that allow
to model optimistic and pessimistic manipulators.
In \autoref{sec:coal-man}, we formally define the coalitional manipulation
problem and explore its computational complexity using \ellBloc{} as a showcase.
We refer to our conclusion and \autoref{tab:results} (\autoref{sec:conclusion})
for a detailed overview of our findings.

\section{Preliminaries}
\label{sec:Prelim}
For a positive integer $n$, let $[n] := \{1, 2, \ldots, n\}$.  We use the
toolbox of parameterized complexity~\citep{CFKLMPPS15,DF13,FG06,Nie06} to
analyze the computational complexity of our problems in a fine-grained way. To
this end, we always identify a \emph{parameter}~$\rho$ that is typically a
positive integer. We call a problem parameterized by~$\rho$
\emph{fixed-parameter tractable} (in $\fpt$) if it is solvable in $f(\rho) \cdot
|I|^{O(1)}$ time, where $|I|$ is the size of a given instance encoding, $\rho$
is the value of the parameter, and $f$ is an arbitrary computable (typically
super-polynomial) function. To preclude fixed-parameter tractability, we use an
established complexity hierarchy of classes of parameterized problems, $\fpt
\subseteq \wone \subseteq \wtwo \subseteq \cdots \subseteq \xp$. It is widely
believed that all inclusions are proper. The notions of hardness for
parameterized classes are defined through parameterized reductions similar to
classical polynomial-time many-one reductions---in this work, it suffices to
ensure that the value of the parameter in the problem we reduce to depends only
on the value of the parameter of the problem we reduce from. Occasionally, we
use a \emph{combined parameter} $\rho' + \rho''$ which is a more explicit way of
expressing a parameter $\rho = \rho' + \rho''$.

An \emph{election}~$(C,V)$ consists of a set~$C$ of $m$~candidates and a
multiset~$V$ of $n$~votes. Votes are linear orders over~$C$---for example, for
$C=\{c_1,c_2,c_3\}$ we write $c_1 \succ_v c_2 \succ_v c_3$ to express that
candidate~$c_1$ is the most preferred and candidate~$c_3$ is the least preferred
according to vote~$v$. We~write~$\succ$ if the corresponding vote is clear from
the context.

A \emph{multiwinner voting rule}%
\footnote{Some literature use the name \emph{multiwinner voting correspondence}
for what we called multiwinner voting rule. In that case, a multiwinner voting
rule is required to return exactly one set of the desired size. This is usually
achieved by combining a multiwinner voting correspondence with a tie-breaking
rule.} is a function that, given an election~$(C,V)$ and an integer $k\in
[|C|]$, outputs a family of co-winning size-$k$ subsets of $C$ representing the
co-winning \emph{\kelitegroup{}s}. We use \emph{\kelgroup{}} as an abbreviation
for \kelitegroup{}; we also use \emph{\elgroup{}} if the size of an
\elitegroup{} is either not relevant or clear from the context. The reason we do
not use the established term ``committee'' is that in shortlisting applications
``committee'' traditionally rather refers to voters and not to candidates.

We consider \emph{scoring rules}---multiwinner voting rules that assign points
to candidates based on their positions in the votes. By~$\score(c)$, we denote
the total number of points that candidate~$c$ obtains, and we
use~$\score_{V'}(c)$ when restricting the election to a subset~$V'\subset V$ of
voters. A (multiwinner) scoring rule selects a family~$\mathcal{X}$ of
co-winning \kelgroup{}s with the maximum total sum of scores. It holds that $X
\in \mathcal{X}$ if and only if $\forall c \in X, c' \in C \setminus X\colon
\score(c) \ge \score(c')$. We focus on the family of \ellBloc{} multiwinner
voting rules, that is a family of scoring rules that assign, for each vote, one
point to each of the top~$\ell < |C|$ candidates.\footnote{The case where~$\ell$
 coincides with the size~$k$ of the \elgroup{} is typically referred to as
 \Bloc; \lBloc{1} corresponds to \SNTV{} \citep{MPRZ08}. The case where~$1 <
\ell < k$ is also referred to as Limited Vote (or Limited Voting). }

\begin{example} \label{ex:unfairUtil}
Referring back to our introductory example, we have a set \exampleCandidates{} of candidates
and a set $V=\{\vvoteYa, \vvoteYb, \vvoteYc, \voteBa, \voteBb, \voteMa, \voteMb\}$ of voters.
The
voters \vvoteYa{}, \vvoteYb{}, and \vvoteYc{} represent the three junior jury members,
whereas \voteBa{}, \voteBb{} and \voteMa{}, \voteMb{} represent, respectively,
the Beethoven and Mozart enthusiasts among the senior jury members. In the example,
we described a way of manipulating the election by the senior jury members which
leads to selecting two music pieces. There are several ways to
illustrate this manipulation using our model. Below we present one of the
possible sets of casted votes that represents the manipulated election:
\begin{align*}
 \vvoteYa, \vvoteYb, \vvoteYc\colon\quad \candPa \succ \candPb \succ \candBa \succ
 \candBb \succ \candMa \succ \candMb, \\
 \voteBa, \voteBb, \voteMa, \voteMb \colon\quad \candBa \succ \candMa \succ
 \candBb \succ \candMb \succ \candPa \succ \candPb.
\end{align*}
Following the introductory example, we are choosing an \elgroup{} of size~$k=2$.
Using the \Bloc{} multiwinner voting rules (which coincides with our
introductory example), the winning \namedelgroup{2} consist of candidates
\candBa{} and \candMa{}. However, under the \SNTV{} voting rule the
situation would change, and the winners would be \candPa{} and \candBa{}.
\SNTV{} and \Bloc{} alike output a single winning \elgroup{} in this example,
and thus tie-breaking is ineffective.
\end{example}

To select a single \kelgroup{} from the set of co-winning \kelgroup{}s one has
to consider tie-breaking rules.
A \emph{multiwinner tie-breaking rule} is a mapping that, given 
an election and a family of co-winning \kelgroup{}s, outputs a single \kelgroup{}.
Among them, there is a set of natural rules that is of particular interest in
order to model the behavior of manipulative voters.
Indeed, in case of a single manipulator both \emph{pessimistic} tie-breaking
as well as \emph{optimistic} tie-breaking have been considered in addition
to lexicographic and randomized tie-breaking~\citep{MPRZ08,OZE13}.
To model optimistic and pessimistic tie-breaking in a meaningful manner\footnote{We
cannot simply use ordinal preferences:
\citet{OZE13} observed that already in case of a single manipulator one cannot simply set the fixed
lexicographic order of the manipulators' preferences (resp.\ the reverse of it)
over candidates to model
optimistic (resp.\ pessimistic) tie-breaking.
For example, it is a strong restriction to assume that a manipulator would always prefer
its first choice together with its fourth choice towards
its second choice together with its third choice.
It might be that only its first choice is really acceptable (in which case the
assumption is reasonable) or that the first three choices are comparably good
but the fourth choice is absolutely unacceptable (in which case the assumption
is wrong).},
we use the model introduced by \citet{OZE13} in which a manipulative voter~$v$
is described not only by the preference order~$\succ_v$ of the candidates
but also by a utility function $u\colon C \rightarrow \naturals$.
To cover this in the tie-breaking process, \emph{coalition-specific}
tie-breaking rules get---in addition to the original election, the manipulators'
votes, and the co-winning \elitegroup{}s---the manipulators' utility functions in
the input. The formal implementations of these rules and their properties are
discussed in \autoref{sec:tie-breaking}.

\section{Model for Coalitional Manipulation}
\label{sec:model}
In this section, we formally define and explain our model and the respective
variants. To this end, we discuss how we evaluate a \kelgroup{} in terms of
utility for a coalition of manipulators and introduce tie-breaking rules that
model optimistic or pessimistic viewpoints of the manipulators.

\subsection{Evaluating \kelgroup{}s}
\label{sec:variants}
As already discussed in the introduction, one should not extend the model of
coalitional manipulation for single-winner elections to multiwinner elections in
the simplest way (e.g.,\ by assuming that the manipulators agree on some
distinguished candidate or on some distinguished \elgroup{}). Instead, we follow
\citet{MPRZ08} and assume that we are given a utility function over the
candidates for each manipulator and a utility level which, if achieved,
indicates a successful manipulation. \citet{MPRZ08} compute the utility of an
\elgroup{} by summing~up the utility values the manipulator assigns to each
member of~the \elgroup{}.

At first glance, summing up the utility values assigned by each manipulator to
each member of an \elgroup{} seems to be the most natural extension for a
coalition of manipulators. However, this \emph{utilitarian variant} does not
guarantee single manipulators to gain non-zero utility. In extreme cases it
could even happen that some manipulator is worse off compared to voting
sincerely, as demonstrated in
\autoref{ex:manipWorseOff}.
\begin{example} \label{ex:manipWorseOff}
 Consider the election $E=(C,V)$ where \exampleCandidates{} is a set
 of~candidates and  $V=\{v_1,v_2,v_3\}$ is the following multiset of three
 votes:
 \begin{align*}
  v_1, v_2 \colon\quad o_1 \succ o_2 \succ m_1 \succ m_2 \succ b_1 \succ b_2, \\
  v_3\colon\quad m_2 \succ m_1 \succ b_2 \succ b_1 \succ o_1 \succ o_2.
 \end{align*}
 Additionally, consider two
 manipulators, $u_1$ and $u_2$, that report utilities to the candidates as
 depicted in the table below.
 \begin{center}
  \begin{tabular}{c|c c c c c c}
   $u(\cdot)$ & $b_1$ & $b_2$ & $m_1$ & $m_2$ & $o_1$ & $o_2$ \\ \hline
   $u_1$ & $10$ & $5$ & $4$ & $0$ & $0$ & $0$\\ 
   $u_2$ & $1$ & $2$ & $5$ & $7$ & $0$ & $0$
  \end{tabular}
 \end{center}
 Let us analyze the winning \namedelgroup{2} under the \SNTV{} voting rule.
 Observe that if the manipulators vote sincerely, then together they give one
 point to~$b_1$ and one to~$m_2$ (one point from each manipulator). Combining
 the manipulators' votes with the non-manipulative ones, the winning
 \namedelgroup{2} consists of candidates $o_1$ and $m_2$ that both have score
 two; no other candidate has greater or equal score, so tie-breaking is
 unnecessary. The value of such a group is equal to seven according to the
 utilitarian evaluation variant.  Manipulator $u_2$'s utility is seven. However,
 both manipulators can do better by giving their points to candidate $b_1$.
 Then, the winners are candidates $o_1$ and $b_1$, giving the total utility of
 $11$ (according to the utilitarian variant).  Observe that in spite of growth
 of the total utility, the utility value gained by $u_2$, which is one, is lower
 than in the case of sincere voting.
\end{example}

In \autoref{ex:manipWorseOff} manipulator~$u_2$ devotes its satisfaction to the
utilitarian satisfaction of the group of the manipulators; that is, $u_2$ is
worse off voting strategically compared to voting sincerely. Despite this issue,
however, the utilitarian viewpoint can be justified if the manipulators are able
to compensate such losses of utility of some manipulators, for example, by
paying money to each other. For~cases where~manipulators cannot do that, we
introduce two egalitarian evaluation variants. The \emph{(\elgroup-wise)
egalitarian} variant aims at maximizing the minimum satisfaction of the
manipulators with the whole \kelgroup{}. Thus, intuitively, in the egalitarian
variant one is looking for an \elgroup{} that is maximizing satisfaction of the
least satisfied manipulator. The \emph{candidate-wise egalitarian} variant aims
at maximizing the manipulators' satisfaction resulting from the summation of the
minimum satisfactions every single candidate contributes. This variant is
motivated, for instance, by the following scenario. We associate candidates with
investments, manipulators with financial experts and utilities with a measure of
a benefit the experts predict for each investment. The candidate-wise
egalitarian is then modelling a ``guaranteed'' benefit of a selection of
investments. The computed benefit is ``guaranteed'' in a sense that, for each
investment, one counts the benefit predicted by the most pessimistic
expert
We do not distinguish ``candidate-wise utilitarian'' variant since this variant
would be equivalent to the (regular) utilitarian variant.

We formalize the described variants of \kelgroup{} evaluation (for
$r$~manipulators) with~\autoref{def:utility}.
\begin{definition} \label{def:utility}
Given a set of candidates $C$, a \kelgroup{} $S \subseteq C$, $|S|=k$, and a
family of manipulators' utility functions $U=\{u_1, u_2, \dots,
u_\manipnumber\}$ where $u_i \colon C \rightarrow \naturals$, we consider the
following functions:
\begin{itemize}
 \item $\util_{U}(S):=\sum_{u \in U} \sum_{c \in S} u(c)$, \item
 $\egal_{U}(S):=\min_{u \in U} \sum_{c \in S} u(c)$, \item
 $\candegal_{U}(S):=\sum_{c \in S} \min_{u \in U} u(c)$.
\end{itemize}
\end{definition} 
Intuitively, these functions determine the utility of a \kelgroup{} $S$
according to, respectively, the utilitarian and the egalitarian variants of
evaluating~$S$ by a group of \manipnumber~manipulators (identifying manipulators
with their utility functions). We omit subscript~$U$ when $U$ is clear from the
context. To illustrate \autoref{def:utility} we apply it in
\autoref{ex:utilities}.
\begin{example} \label{ex:utilities}
 Consider our example set of candidates \exampleCandidates{} and two
 manipulators $u_1$, $u_2$ whose utility functions over the candidates are
 depicted in the table below.
 \begin{center}
  \begin{tabular}{c|c c c c c c}
   $u(\cdot)$ & $b_1$ & $b_2$ & $m_1$ & $m_2$ & $o_1$ & $o_2$ \\ \hline $u_1$ &
   $10$ & $5$ & $4$ & $0$ & $0$ & $0$\\ $u_2$ & $1$ & $2$ & $5$ & $7$ & $0$ &
   $0$
  \end{tabular}
 \end{center}
 Then, evaluating the utility of \namedelgroup{2} $S=\{b_1, m_1\}$ applying the
 three different evaluation variants gives:
 \begin{itemize}
  \item $\util(S)=(10+4) + (1+5) = 20$, \item $\egal(S)=\min\{(10+4);(1+5)\}=6$,
 \item $\candegal(S)=\min\{10,1\} + \min\{4, 5\} = 5$.
 \end{itemize}
\end{example}

Analyzing \autoref{ex:utilities}, we observe that we can compute the utilitarian
value of \elgroup{}~$S$ by summing up the overall utilities that each candidate
in~$S$ contributes to all manipulators; for instance candidate~$b_1$ always
contributes the utility of~$11=10+1$ to the manipulators, independently of other
candidates in the~\elgroup{}. Following this observation, instead of coping with
a collection of utility function, we can ``contract'' all manipulator's
functions to a single function. The new function assigns each candidate a
utility value equal to the sum of utilities that the contracted functions assign
to this candidate. Analogously, we can deal with the candidate-wise egalitarian
variant by taking the minimum utility associated to each candidate as the the
utility of this candidate in a new function. Thus, in both variants, we can
consider a single utility function instead of a family of functions. The
conclusion from the above discussion is summarized in the following observation.

\begin{observation} \label{obs:util-functions}
 Without loss of generality, one can assume that there is a single non-zero
 valued utility function over the candidates under the utilitarian or
 candidate-wise egalitarian evaluation.
\end{observation}
\begin{proof}
 Consider a multiset of manipulators' utility functions $U=\{u_1, u_2,
 \ldots, u_\manipnumber\}$ and a \kelgroup{} $S=\{c_1, c_2, \ldots, c_k\}$.
 For the utilitarian variant, create a new utility function~$u'$ that assigns to
 each candidate the sum of utilities given to this candidate by all
 manipulators;
 that is, $u'(c):=\sum_{i=1}^r u_r(c)$ for all $c \in C$.
 Technically, we need also a special utility function $u_0$
 that assigns to each candidate the utility of zero.
 We construct a new family
 $U'$ consisting of function $u'$ and $\manipnumber - 1$ copies of
 function~$u_0$. Since function $u'$ is the only one which gives non-zero
 utility in family $U'$ and for each candidate this function returns the sum of
 utilities given to a candidate by all functions from family $U$, it holds that
 $\util_{U}(S)=\util_{U'}(S)$.  

 We follow a similar strategy proving \autoref{obs:util-functions} for the
 candidate-wise variant. We introduce a function~$u'$ defined as~$u'(c):=\min_{u
 \in U}u(c)$ for each candidate~$c \in C$. Then we create a new family of
 utility functions~$U'$ with function~$u'$ and $\manipnumber -1$ zero-valued
 function~$u_0$. Naturally, $\egal_{U'}(S)=\egal_{U}(S)$ because the values of
 function~$u'$ for each candidate exactly follow the definition of the
 candidate-wise egalitarian evaluation (see \autoref{def:utility})  and $u'$ is
 the only non-zero valued function in $U'$.
\end{proof}

\subsection{Breaking Ties}
\label{sec:tie-breaking}

Before formally defining our tie-breaking rules, we briefly discuss some
necessary notation and central concepts. Consider an election $(C,V)$, a
size~$k$ for the \elgroup{} to be chosen, and a scoring-based multiwinner voting
rule~$\calR$. We can partition the set of candidates~$C$
into three sets \Cconf{}, \Cpend{}, and \Crej{} as follows: The set~\Cconf{}
contains the \emph{confirmed candidates}, that is, candidates that are in all
co-winning \kelgroup{}s. The set~\Cpend{} contains the \emph{pending
candidates}, that is, candidates that are only in some co-winning \kelgroup{}s.
The~set~\Crej{} contains the \emph{rejected candidates}, that is, candidates
that are in no co-winning \kelgroup{}. Observe that $|\Cconf| \le k$,
$|\Cconf\cup\Cpend|\ge k$, and that every candidate from~$\Cpend \cup \Crej$
receives fewer points than every candidate from~\Cconf{}. Additionally, all
candidates in \Cpend{} receive the same number of points.

We define the following families of tie-breaking rules which are considered in
this work. In order to define optimistic and pessimistic rules, we assume that
in addition to~\Cconf{}, \Cpend{}, and $k$, we are given a family of utility
functions which are used to evaluate the \kelgroup{}s as discussed in
\autoref{sec:variants}.

 \textbf{Lexicographic \Flex.}
 A tie-breaking~$\calF$ belongs to $\Flex$ if and only if ties are broken
 lexicographically with respect to some predefined order~$\lexpref$ of the
 candidates from~$C$. That is, $\calF$~selects all candidates from~\Cconf{} and
 the top $k-|\Cconf|$ candidates from~$P$ with respect to~$\lexpref$.

 \textbf{Optimistic $\Fopt^{\eval}$, $\eval \in \{\util,\egal,\candegal\}$.}
 A tie-breaking belongs to $\Fopt^{\eval}$ if and only if
 it always selects some \kelgroup~$S$ such that $\Cconf
 \subseteq S \subseteq (\Cconf \cup P)$ and~there is no other \kelgroup~$S'$
 with $\Cconf \subseteq S' \subseteq (\Cconf \cup P)$ and $\eval(S') >
 \eval(S)$.

 \textbf{Pessimistic $\Fpess^{\eval}$, $\eval \in \{\util,\egal,\candegal\}$.}
 A tie-breaking belongs to $\Fpess^{\eval}$ if and only if
 it always selects some \kelgroup~$S$ such that $\Cconf
 \subseteq S \subseteq (\Cconf \cup P)$ and~there is no other \kelgroup~$S'$
 with $\Cconf \subseteq S' \subseteq (\Cconf \cup P)$ and $\eval(S') <
 \eval(S)$.

We remark that the definitions above come in two, substantially different
variants. For each lexicographic tie-breaking rule, there is always exactly one
\elgroup{} that will be selected by the rule for a particular set of pending
set candidates. However, it is not the case for the families of pessimistic and
optimistic families of rules. In fact, there might be many possible \elgroup{}s
whose value, computed in terms of a respective evaluation variant, is exactly
the same. Such a feature seems to contradict the idea of a tie-breaking rule
that should not, by itself, introduce ties again. However, we argue that
choosing arbitrary equally-valued (``tied'') \elgroup{} is a proper way to
circumvent this problem. Indeed, according to a particular evaluation all
\elgroup{}s with the same value are indistinguishable from each other.

\subsection{Limits of Lexicographic Tie-Breaking}
\label{subsec:limits}
\newcommand{\fixprefix}[1]{\ensuremath{\{#1\}}}
\newcommand{\fixsetprefix}[1]{\ensuremath{#1}}

\newcommand{\fixsimulation}[1]{\fixprefix{#1}-simulation}
\newcommand{\fixsimulate}[1]{\fixprefix{#1}-simulate}
\newcommand{\fixsimulates}[1]{\fixsimulate{#1}s}
\newcommand{\fixsetsimulation}[1]{\fixsetprefix{#1}-simulation}
\newcommand{\fixsetsimulate}[1]{\fixsetprefix{#1}-simulate}
\newcommand{\fixsetsimulates}[1]{\fixsetsimulate{#1}s}

\newcommand{\fixsimulator}[1]{\fixprefix{#1}-simulator}
\newcommand{\fixsetsimulator}[1]{\fixsetprefix{#1}-simulator}
From the above discussion, we can conclude that lexicographic tie-breaking is
straightforward in~the~case~of~scoring-based multiwinner voting
rules.
Basically any~subset of the desired cardinality from the set
of pending candidates can be chosen.
In particular, the~best pending candidates with respect to the given order
can~be~chosen.
We remark that applying lexicographic tie-breaking may be more
complicated for general multiwinner voting rules.

It remains to be clarified whether one can find a reasonable order of the
pending candidates in order to model optimistic or pessimistic tie-breaking
rules in a simple way. We show that this is possible for every
$\calF^{\eval}_{\bhav}$, $\eval \in \{\util,\candegal\}$, $\bhav \in
\{\opt,\pess\}$, using the fact that in these cases we can safely assume that
there is only one non-zero valued utility function (see
\autoref{obs:util-functions}). On the contrary, there is a counterexample for
$\eval=\egal$ and $\bhav \in \{\opt,\pess\}$. On the way to prove these claims we
need to formally define what it means that one family of tie-breaking rules can
be used to simulate another family of tie-breaking rules.
\begin{definition} \label{def:simulation}
 We call a triplet consisting of an election with candidate set~$C$, a size of an
 \elgroup{}, and a family of utility functions a \emph{tie-breaking perspective
 over~$C$}. Let~$\mathbb{E}_C(\Cconf)$, $\mathbb{E}_C(\Cpend)$,
 $\mathbb{E}_C(U)$, and $\mathbb{E}_C(k)$ be all possible tie-breaking
 perspectives over~$C$ admitting, respectively, a set~\Cconf{} of
 confirmed candidates, a set~\Cpend{} of pending candidates, a family~$U$ of
 utility functions, and a size~$k$ of an \elgroup. For some non-empty
 subset~$\mathcal{P}$ of the set $\{\Cconf{},\Cpend{}, U, k\}$,
 let~$\mathbb{E}_C(\mathcal{P}):=\bigcap_{x \in \mathcal{P}} \mathbb{E}_C(x)$.\\
 Then, for two families~$\calF$ and~$\calF'$ of tie-breaking rules we say that
 $\calF$ can \emph{\fixsetsimulate{\mathcal{P}}$~\calF'$} if there exists a
 rule~$F \in \calF$ such that for each tie-breaking perspective
 in~$\mathbb{E}_C(\mathcal{P})$ there exists a rule~$F' \in \calF'$ such that
 $F$ and~$F'$ yield the same output for this perspective. We call rule~$F$
 a~\emph{\fixsetsimulator{\mathcal{P}}}.
\end{definition}

At first glance, \autoref{def:simulation} might seem overcomplicated. However,
it is tailored to grasp different degrees of simulation possibilities. On the one
hand, one can always find a lexicographic order and use it for breaking ties if
all: confirmed candidates, pending candidates, utility functions, and the size
of an \elgroup{} are known. Thus, one needs some flexibility in the definition
of simulation for it to be non-trivial.
On the other hand, it is somewhat obvious that without fixing the utility functions,
one cannot simulate optimistic or pessimistic tie-breaking rules.
In other words, we have:

\begin{observation}\label{obs:lex-simulation-withoutU}
 Let $C$~be a fixed set of candidates, \Cconf{} be a set of confirmed
 candidates, \Cpend{} be a set of pending candidates, and $k$ be a size of an
 \elgroup{}. Let $\bhav \in \{\opt,\pess\}$ and $\eval \in \{\util,\candegal,
 \egal\}$. The family of lexicographic tie-breaking rules does not
 \fixsimulate{\Cconf, \Cpend, k}~$\calF^{\eval}_{\bhav}$.
\end{observation}
\begin{proof}
 Suppose $k=1$, $\Cconf=\emptyset$, and $\Cpend=\{\candBa, \candBb\}$;
 that is, we are going to select either \candBa{} or \candBb{} who are tied.
 Let us fix a family $U=\{u_1\}$ of utility functions such that $u_1(\candBa)=1$ and
 $u_1(\candBb)=0$.
 For the family~$U$ of utility functions clearly $\calF^{\eval}_{\opt}$ selects
 candidate \candBa{}.
 Now, consider a family $U'=\{u'_1\}$ of utility functions where $u'_1$ assigns
 utility one to candidate \candBb{} and zero otherwise.
 For this family, $\calF^{\eval}_{\opt}$ selects candidate \candBb.  This means
 that we cannot find a~\fixsimulator{\Cconf, \Cpend, k}~$F$ from family $\Flex$
 of tie-breaking rules because in the first case $F$ would have to
 choose~\candBa{} and in the second case \candBb{} would have to be chosen.
 This is impossible using a single preference order over~$\{\candBa, \candBb\}$.
 Similar families of functions (obtained by exchanging each one with zero and
 vice versa) yield a proof for~$\calF^{\eval}_{\pess}$ as well.
\end{proof}

Next, we show that for some cases it is sufficient to fix just the
utility functions in order to simulate optimistic or pessimistic tie-breaking rules
(see~\autoref{prop:lex-simulation-util}).
For other cases, however, one \emph{has} to fix all: confirmed candidates, pending candidates,
utility functions, and the size of an \elgroup{} (see \autoref{prop:lex-simulation-egal}).

\begin{proposition}\label{prop:lex-simulation-util}
 Let $C$~be a set of candidates, $U$~be a family of utility functions, $\bhav
 \in \{\opt,\pess\}$, and $\eval \in \{\util,\candegal\}$. Let~$|C|=m$ and
 $|U|=r$. Then the family of lexicographic tie-breaking rules~\Flex{} can
 \fixsimulate{U}~$\calF^{\eval}_{\bhav}$, and a~\fixsimulator{U} $F \in \Flex$ can be
 found in $O(m \cdot (\manipnumber + \log m))$~time.
\end{proposition}

\begin{proof}
 Recall from \autoref{obs:util-functions} that if $\eval \in \{\util,\candegal\}$,
 then there is always a set of utility functions with just one non-zero valued
 utility function~$u'$ that is equivalent to~$U$.
 Hence, we compute such a function~$u'$ in $O(m \cdot \manipnumber)$ time as follows:
 In the utilitarian case, function $u'$ assigns every candidate the sum of
 utilities the manipulators give to the candidate.
 Considering the candidate-wise egalitarian evaluation, function $u'$ assigns every
 candidate the minimum utility value among utilities given to the candidate over all
 manipulators.
 We say an order~$\lexpref$ of the~candidates is \emph{consistent}
 with some utility function~$u$ if $c \lexpref c'$ implies $u(c) \ge u(c')$ for
 optimistic tie-breaking and $c \lexpref c'$ implies $u(c) \le u(c')$ for
 pessimistic tie-breaking.
 Any lexicographic tie-breaking rule defined by an
 order~$\lexpref$ that is consistent with~the utility function~$u'$ simulates
 $\calF^{\eval}_{\bhav}$.
 We compute a~consistent order by sorting the
 candidates according to~$u'$ in $O(m \cdot \log m)$ time.         
\end{proof}

\autoref{prop:lex-simulation-util} describes a strong feature of optimistic
utilitarian and candidate-wise egalitarian tie-breaking and their pessimistic
variants. Intuitively, the proposition says that for these tie-breaking
mechanisms one can compute a respective linear order of candidates. Then one can
forget all the details of the initial tie-breaking mechanism and use the~order
to determine winners. The~order can be computed even without knowing the details
of an election.
Unfortunately, the simulation of pessimistic and optimistic egalitarian
tie-breaking turns out to be more complicated.

\begin{proposition}\label{prop:lex-simulation-egal}
 Let $C$~be a set of candidates, $U$~be a family of utility functions, \Cconf{}
 be a set of confirmed candidates, \Cpend{} be a set of pending candidates, and
 $k$ be a size of an \elgroup{}. For each $\mathcal{P} \subseteq \{\Cconf,
 \Cpend, U, k\}$, $0<|\mathcal{P}|<4$, the lexicographic tie-breaking family of
 rules does not \fixsetsimulate{\mathcal{P}}~$\calF^{\egal}_{\bhav}$ assuming
 $\bhav \in \{\opt,\pess\}$.
\end{proposition}
\begin{proof}
 From \autoref{obs:lex-simulation-withoutU} we already know that the family of
 lexicographic tie-breaking rules cannot \fixsimulate{\Cconf, \Cpend, k}
 the family of egalitarian pessimistic tie-breaking rules or the family of
 egalitarian optimistic tie-breaking rules.

 Next, we build one counterexample for each of the remaining size-three subsets
 of~$\{\Cconf, \Cpend, U, k\}$ to show our theorem.
 To this end, let us fix a set of candidates~$C=\{b_1, b_2, m_1, m_2, o_1, o_2\}$
 (compatible with our running example) and a family~$U=\{u_1, u_2\}$ of utility
 functions as depicted in the table below.
 \begin{center}
  \begin{tabular}{c|c c c c c c}
   $u(\cdot)$ & $b_1$ & $b_2$ & $m_1$ & $m_2$ & $o_1$ & $o_2$ \\ \hline
   $u_1$ & $10$ & $5$ & $4$ & $0$ & $0$ & $0$\\ 
   $u_2$ & $1$ & $2$ & $5$ & $7$ & $0$ & $0$
  \end{tabular}
 \end{center}

 First, we prove that the family \Flex{} cannot \fixsimulate{\Cconf, \Cpend,
 U}~$\calF^{\egal}_{\bhav}$ for $\bhav \in \{\opt,
 \pess\}$.  Let us fix $\Cconf=\emptyset$, $\Cpend=C \setminus \{o_1, o_2\}$. We
 consider the optimistic variant of egalitarian tie-breaking for~$k = 1$, so we are
 searching for a \namedelgroup{1}. Looking at the values~of~$U$,
 we see that candidate~$m_1$ gives the best possible egalitarian evaluation
 value which is four. This means that a~\fixsimulator{\Cconf, \Cpend, U}~$F \in
 \Flex$ has to use an order where $m_1$ precedes both $b_1$ and $m_2$. However,
 it turns out that if we set $k=2$, then the best \namedelgroup{2} consists
 exactly of candidates $b_1$ and $m_2$.  This leads to a contradiction because
 now candidates $b_1$ and $m_2$ should precede $m_1$ in $F$'s lexicographic
 order. Consequently, family~$\Flex$ does not \fixsimulate{\Cconf, \Cpend,
 U}~$\calF^{\egal}_{\opt}$.  Using the same values of utility functions and the
 same sequence of the values of $k$ we get a proof for the pessimistic variant
 of egalitarian evaluation.

 Second, we prove that the family \Flex{} cannot \fixsimulate{\Cpend, k,
 U}~$\calF^{\egal}_{\bhav}$ for $\bhav \in
 \{\opt, \pess\}$. This time, we fix $\Cpend=C\setminus\{o_1, o_2\}$, $k=2$. We
 construct the first case by setting $\Cconf=\{o_1\}$. Using the fact that in
 both functions candidate $o_1$ has utility zero, we choose exactly the same
 candidate as in the proof of \fixsimulation{\Cconf, \Cpend, U} for the case
 $k=1$; that is, for the optimistic variant, the winning \namedelgroup{2}
 is~$m_1$ and~$o_1$. Consequently, $m_1$ precedes~$b_1$ and~$m_2$ in the
 potential \fixsimulator{\Cpend, k, U}'s lexicographic order. Towards a
 contradiction, we set $\Cconf=\emptyset$. The situation is exactly the same as
 in the proof of the \fixsimulation{\Cconf, \Cpend, U} case. Now, the winning
 \namedelgroup{2} consists of $b_1$ and $m_2$ which ends the proof for the
 optimistic case. By almost the same argument, the result holds for the
 pessimistic variant.

 Finally, we prove that the family \Flex{} cannot \fixsimulate{\Cconf, k,
 U}~$\calF^{\egal}_{\bhav}$ for $\bhav \in
 \{\opt, \pess\}$.  We fix $\Cconf=\emptyset$, $k=2$. For the first case we pick
 $\Cpend=\{b_2, m_1, m_2\}$. The best egalitarian evaluation happens for the
 \namedelgroup{2} consisting of $b_2$ and $m_1$. This imposes that, in the
 potential \fixsimulator{\Cconf, k, U}'s order, $b_2$ and~$m_1$ precede the
 remaining candidates (in particular, $m_1$ precedes $m_2$).  However, for
 $\Cpend=C$ the best \namedelgroup{2} changes to that consisting of $b_1$ and
 $m_2$ which gives a contradiction ($m_2$~precedes $m_1$).  As in the previous
 cases, the same argument provides a proof for the pessimistic variant.
\end{proof}

\autoref{prop:lex-simulation-egal} implies that pessimistic and optimistic
egalitarian tie-breaking cannot be simulated without having full knowledge about
an election. In terms of computational complexity, however, finding winners for
pessimistic egalitarian tie-breaking remains tractable whereas the same task for
optimistic egalitarian tie-breaking is intractable. We devote the next section
to show this dichotomy as well as to establish computational hardness of
computing winners for the other introduced tie-breaking rules.

\section{Complexity of Tie-Breaking}
\label{sec:compl:tie-breaking}

It is natural to ask whether the tie-breaking rules proposed
in~\autoref{sec:tie-breaking} are practical in terms of their computational
complexity. If not, then there is little hope for coalitional manipulation
because tie-breaking might be an inevitable subtask to be solved by the
manipulators. Indeed, manipulators might not be ``powerful'' enough to secure
victory of their desired \elgroup{} completely avoding tie-breaking.

Clearly, we can apply every lexicographic tie-breaking rule that is
defined through some predefined order of the candidates in linear time.
Hence, we focus on the rules that model optimistic or pessimistic manipulators.
To this end, we analyze the following computational problem.

\probDefLong[$\eval \in \{\util, \egal, \candegal\}$, $\bhav \in \{\opt,\pess\}$]
{\TIElong{\eval}{\bhav}}{\TIE{\eval}{\bhav}}
{A set of candidates~$C$ partitioned into
 a set~$\Cpend$ of pending candidates and a set~$\Cconf$ of confirmed
 candidates, the size~$k$ of the \elitegroup{} such that $|\Cconf|<k<|C|$, a
 family of manipulators' utility functions $U=\{u_1, u_2, \dots,
 u_\manipnumber\}$ where $u_i \colon C \rightarrow \naturals$, and a
 non-negative, integral evaluation threshold~$q$.}
{Is there a size-$k$ set $S \subseteq C$ such that $S$~is selected according
 to $\calF^{\eval}_{\bhav}$, $\Cconf \subseteq S$, and $\eval(S) \ge q$?}

Naturally, we may assume that the number of candidates and the number of utility
functions are polynomially bounded in the size of the input. However, both the
evaluation threshold and the utility function values are encoded in binary.

Note that an analogous problem has not been considered for single-winner
elections. The reason behind this is that, for single-winner elections,
optimistic and pessimistic tie-breaking rules can be easily simulated by
lexicographic tie-breaking rules. To obtain them, it is sufficient to order the
candidates with respect to their value to manipulators, computed separately for
every candidate. However, one cannot simply apply this approach for \elgroup{}s,
because there might be exponentially many different \elgroup{}s to consider.
Even if this exponential blow-up were acceptable, it would be unclear how to
derive an order of candidates from the computed values of \elgroup{}s. Yet,
using a different technique, we can simulate tie-breaking in multiwinner
elections with a lexicographic tie-breaking rule for several variants of
evaluation.

\subsection{Utilitarian and Candidate-Wise Egalitarian: Tie-Breaking Is Easy}
As a warm-up, we observe that tie-breaking can be applied and evaluated
efficiently if the \kelgroup{}s are evaluated according to the utilitarian or
candidate-wise egalitarian variant. The corresponding result follows almost
directly from \autoref{prop:lex-simulation-util}.

\begin{corollary}\label{cor:TBeasy}
 Let $m$ denote the number of candidates and $\manipnumber$ denote the number of
 manipulators. Then one can solve \TIElong{\eval}{\bhav} in $O(m \cdot
 (\manipnumber + \log m))$ time for $\eval \in \{\util,\candegal\}$, $\bhav \in
 \{\opt,\pess\}$.
\end{corollary}

\begin{proof}
 The algorithm works in two steps. First, compute a~lexicographic tie-breaking
 rule~$\Flex$ that simulates $\calF^{\eval}_{\bhav}$ in $O(m \cdot (\manipnumber
 + \log m))$ time as described in \autoref{prop:lex-simulation-util}. Second,
 apply tie-breaking rule $\Flex$, and evaluate the resulting \kelgroup{} in $O(k
 \cdot \manipnumber)$ time. The running time of applying a lexicographic
 tie-breaking rule is linear with respect to the input length (see
 \autoref{subsec:limits}).
\end{proof}

\subsection{Egalitarian: Being Optimistic Is Hard}
\label{sec:hardTB}

In this subsection, we consider the optimistic and pessimistic
tie-breaking rules when applied for searching a \kelgroup{} evaluated according
to the egalitarian variant. First, we show that applying and evaluating
egalitarian tie-breaking is computationally easy for pessimistic manipulators
but computationally intractable for optimistic manipulators even if the size of
the \elgroup{} is small. Being pessimistic, the main idea is to ``guess'' the
manipulator that is least satisfied and select the candidates appropriately. We
show the computational worst-case hardness of the optimistic case via a
reduction from the \wtwo{}-complete \textsc{Set Cover} problem parameterized by
solution size~\citep{DF13}.

\newcommand{\SCuniverse}{\ensuremath{X}}
\newcommand{\SCunivelement}{\ensuremath{x}}

\begin{theorem}\label{thm:egalTB}
 Let $m$ denote the number of candidates, $\manipnumber$ denote the number of
 manipulators, $q$ denote the evaluation threshold, and $k$ denote the size of
 an \elgroup{}. Then one can solve \TIElong{\egal}{\pess} in $O(\manipnumber
 \cdot m \log m)$ time, but \TIElong{\egal}{\opt} is \np-hard and \wtwo-hard
 when parameterized by~$k$ even if $q=1$ and every manipulator only gives either
 utility one or zero to each candidate.
\end{theorem}
\begin{proof}
 For the pessimistic case, it is sufficient to ``guess'' the least satisfied
 manipulator~$x$ by iterating through \manipnumber{} possibilities. Then, select
 $k-|\Cconf|$~pending candidates with the smallest total utility for this
 manipulator in $O(m \log m)$ time. Finally, comparing the \kelgroup{} with the
 worst minimum satisfaction over all manipulators to the lower bound $q$ on
 satisfaction level given in the input solves the problem.
 
 We prove the hardness for the optimistic case reducing from the $\wtwo$-hard
 \textsc{Set Cover} problem which, given a collection $\mathcal{S}=\{S_1, S_2,
 \dots, S_m\}$ of subsets of universe
 \namedorderedsetof{\SCuniverse}{\SCunivelement}{n} and
 an integer $h$, asks whether there exists a family $\mathcal{S'} \subseteq
 \mathcal{S}$ of size at most~$h$ such that $\bigcup_{S \in \mathcal{S'}} S =
 \SCuniverse$. Let us fix an instance $I=(\SCuniverse, \mathcal{S}, h)$ of
 \textsc{Set Cover}. To construct an \TIElong{\egal}{\opt} instance, we
 introduce pending candidates $\Cpend=\{c_1, c_2, \ldots, c_m\}$ representing
 subsets in $\mathcal{S}$ and manipulators $u_1, u_2, \ldots, u_n$ representing
 elements of the universe. Note that there are no confirmed and rejected
 candidates. Each manipulator $u_i$ gives utility one to candidate~$c_j$ if
 set~$S_j$ contains element~$\SCunivelement_i$ and zero otherwise. We set the
 \elitegroup{} size $k:=h$ and the threshold~$q$ to be $1$.
 
 Observe that if there is a size-$k$ subset $S \subseteq \Cpend$ such that
 ${\min_{i \in [n]} \sum_{s \in S} u_i(s) \geq 1}$, then there exists a family
 $\mathcal{S}'$---consisting of the sets represented by candidates in~$S$---such
 that each element of the universe belongs to the set $\bigcup_{S \in
 \mathcal{S'}} S$. On the contrary, if we cannot pick a group of candidates of
 size $k$ for which every manipulator's utility is at least one, then instance
 $I$ is a `no' instance. This follows from the fact that for each size-$k$
 subset $S \subseteq \Cpend$ there exists at least one manipulator $u^*$ for
 whom ${\sum_{s \in S} u^*(s) = 0}$.  This translates to the claim that there
 exists no size-$h$ subset $\mathcal{S}' \subseteq \mathcal{S}$ such that all
 elements in~$\SCuniverse$ belong to the union of the sets in $\mathcal{S}'$.
 
 Since \textsc{Set Cover} is $\np$-hard and $\wtwo$-hard with
 respect to parameter~$h$, we obtain that our problem is also $\np$-hard and
 $\wtwo$-hard when parameterized by the size~$k$ of an \elitegroup{}.
\end{proof}

\newcommand{\MCC}{\textsc{Multicolored Clique}}

Inspecting the \wtwo{}-hardness proof of \autoref{thm:egalTB}, we learn that
a small \elgroup{} size (alone) does not make \TIElong{\egal}{\opt}
computationally tractable even for very simple utility functions. Next, using a
parameterized reduction from the \wone{}-complete~\MCC{} problem~\citep{FHRV09},
we show that there is still no hope for fixed-parameter tractability (under
standard assumptions) even for the combined parameter ``number of
manipulators and \elgroup{} size''; intuitively, this parameter covers
situations where few manipulators are going to influence an election for a small
\elgroup{}.

\begin{theorem}\label{thm:egalTBcomb}
 Let $k$ denote the size of an \elgroup{} and $\manipnumber$ denote the number of
 manipulators. Then, parameterized by~$\manipnumber+k$, \TIElong{\egal}{\opt} is
 \wone{}-hard.
\end{theorem}

\begin{proof}
 We describe a parameterized reduction from the \wone{}-hard \MCC{}
 problem~\citep{FHRV09}. In this problem, given an undirected graph~$G=(V,E)$, a
 non-negative integer~$h$, and a vertex coloring~$\phi \colon V\to \{1, 2,
 \ldots, h\}$, we ask whether graph~$G$ admits a colorful $h$-clique, that is, a
 size-$h$ vertex subset~$Q\subseteq V$ such that the vertices in $Q$ are
 pairwise adjacent and have pairwise distinct colors. Without loss of
 generality, we assume that the number of vertices of each color is the same; to
 be referred as~$y$ in the following. Let $(G,\phi)$, $G=(V,E)$, be a \MCC{}
 instance. Let $V(i) = \{v_1^{i}, v_2^{i}, \ldots, v_y^{i}\}$ denote the set of
 vertices of color~$i \in [h]$, and let $E(i,j) = \{e_1^{i,j}, e_2^{i,j},
 \ldots, e_{|E(i,j)|}^{i,j}\}$, defined for $i, j \in [h]$, $i<j$, denote the
 set of edges that connect a vertex of color~$i$ to a vertex of color~$j$.
 
 \mypar{Candidates.}
 We create one confirmed candidate~$c^*$ and~$|V|+|E|$ pending candidates. More
 precisely: for each $\ell \in [y]$, we create one \emph{vertex
 candidate}~$a_\ell^{i}$ for each vertex~$v_\ell^{i} \in V(i)$, $i \in [h]$ and
 for each $i,j \in [h]$ such that $i<j$ we create one \emph{edge
 candidate}~$b_t^{i,j}$ for each edge ~$e_t^{i,j} \in E(i,j)$, $t \in [E(i,j)]$.
 We set the size~$k$ of the \elgroup{} to $h+\binom{h}{2}+1$ and set the
 evaluation threshold~$q:=y+1$. Next, we describe the manipulators and explain
 the high-level idea of the construction.
 
 \mypar{Manipulators and main idea.} Our construction will ensure that there is a
 \kelgroup~$X$ with $c^* \in X$ and $\egal(X) \geq q$ if and only if
 $X$~contains $h$~vertex candidates and $\binom{h}{2}$~edge candidates that
 encode a colorful $h$-clique. To this end, we introduce the following
 manipulators.
 
 \begin{enumerate}
  \item \label{enum:colorman}
   For each color~$i \in [h]$, there is a \emph{color manipulator}~$\mu_i$
   ensuring that the \kelgroup{} contains a vertex candidate~$a_{z_i}^{i}$
   corresponding to a vertex of color~$i$. Herein, variable~$z_i$~denotes the id
   of the vertex candidate (resp.\ vertex) that is \emph{selected} for
   color~$i$.
  \item 
   For each $i,j \in [h]$ such that $i<j$, there is one \emph{color pair
   manipulator}~$\mu_{i,j}$ ensuring that the \kelgroup{} contains an edge
   candidate~$b_{z_{i,j}}^{i,j}$ corresponding to an edge connecting vertices of
   colors $i$ and $j$. Herein, variable $z_{i,j}$~denotes the id of the edge
   candidate (resp.\ edge) that is \emph{selected} for color pair~$\{i,j\}$,
   $i<j$.
  \item \label{enum:verifman} 
   For each $i,j \in [h]$ such that $i \neq j$, there are two
   \emph{verification manipulators}~$\nu_{i,j}$, $\nu'_{i,j}$ ensuring that
   vertex~$v_{z_i}^{i}$ is incident to edge~$e_{z_{i,j}}^{i,j}$ if $i<j$ or
   incident to edge~$e_{z_{j,i}}^{j,i}$ otherwise.
 \end{enumerate}
 
 If there exists a \kelgroup{} in agreement with the description in the
 previous three points, then this \kelgroup{} must encode a colorful $h$-clique.
 
 \mypar{Utility functions.}
 Let us now describe how we can guarantee correct roles of the manipulators
 introduced in points \ref{enum:colorman} to \ref{enum:verifman} above using
 utility functions.
 \begin{enumerate}
  \item Color manipulator~$\mu_i$, $i\in[h]$, has utility~$y$ for the confirmed
   candidate~$c^*$, utility one for each candidate corresponding to a vertex of
   color~$i$, and utility zero for the remaining candidates.
  \item Color pair manipulator~$\mu_{i,j}$, $i,j \in [h]$, $i<j$, has
   utility~$y$ for the confirmed candidate~$c^*$, utility one for each candidate
   corresponding to an edge connecting a vertices of colors $i$ and~$j$, and
   utility zero for the remaining candidates.
 \item Verification manipulator~$\nu_{i,j}$, $i,j \in [h]$, $i\neq j$, has
  utility~$\ell$ for candidate~$a^i_{\ell}$, $\ell \in [y]$, utility~$q-\ell$
  for each candidate corresponding to an edge that connects vertex $v^i_{\ell}$
  to a vertex of color~$j$, and utility zero for the remaining candidates.
 \item Verification manipulator~$\nu'_{i,j}$, $i,j \in [h]$, $i\neq j$, has
  utility~$q-\ell$ for candidate~$a^i_{\ell}$, $\ell \in [y]$, utility~$\ell$
  for each candidate corresponding to an edge that connects vertex~$v^i_{\ell}$,
  to a vertex of color~$j$, and utility zero for the remaining candidates.
 \end{enumerate}
 
 \mypar{Correctness.}
 We argue that the graph $G$ admits a colorful clique of size $h$ if
 and only if there is a \kelgroup~$X$ with $c^* \in X$ and $\egal(X) \geq q$.
 
 Suppose that there exists a colorful clique~$H$ of size $h$. Create the
 \kelgroup~$X$ as follows. Start with~$\{c^*\}$ and add every vertex candidate
 that corresponds to some vertex of~$H$ and every edge candidate that corresponds
 to some edge of~$H$. Each color manipulator and color pair manipulator receives
 total utility~$y+1$, because $H$~contains, by definition, one vertex of each
 color and one edge connecting two vertices for each color pair. It is easy to
 verify that the verification manipulator~$\nu_{i,j}$ must receive
 utility~$\ell$ from a vertex candidate and utility~$q-\ell$ from an edge
 candidate and that the verification manipulator~$\nu'_{i,j}$ must receive
 utility~$q-\ell$ from a vertex candidate and utility~$\ell$ from an edge
 candidate. Thus, $\egal(X)=q=y+1$.
 
 Suppose that there exists a \kelgroup~$X\subseteq C$ such that $\egal(X) \geq
 q$. Since each color manipulator cannot achieve utility $y+1$ unless $c^*$
 belongs to the winning \kelgroup{}, it follows that $c^* \in X$. Because each
 color manipulator~$\mu_i$ receives total utility at least~$y+1$, $X$ must
 contain some vertex candidate~$a_{z_i}^{i}$ corresponding to a vertex of
 color~$i$ for some~$z_i \in [y]$. We say that $X$~\emph{selects}
 vertex~$v_{z_i}^{i}$. Since each color pair manipulator~$\mu_{i,j}$ receives
 total utility at least~$y+1$, $X$ must contain some edge
 candidate~$b_{z_{i,j}}^{i,j}$ corresponding to an edge connecting a vertex of
 color~$i$ and a vertex of color~$j$ for some~$z_{i,j}$. We say that
 $X$~\emph{selects} edge~$e_{z_{i,j}}^{i,j}$. We implicitly assumed that each
 color manipulator and color pair manipulator contributes exactly one selected
 candidate to $X$. This assumption is true because there are exactly $k-1$ such
 manipulators and each needs to select at least one candidate; hence, $X$~is
 exactly of the desired size. In order to show that the corresponding vertices
 and edges encode a colorful $h$-clique, it remains to show that no selected
 edge is incident to a vertex that is not selected. Assume towards a
 contradiction that $X$~selects an edge $e_{z_{i,j}}^{i,j}$ and some
 vertex~$v_{z_i}^{i} \notin e_{z_{i,j}}^{i,j}$. However, either verification
 manipulator~$\nu_{i,j}$ or verification manipulator~$\nu'_{i,j}$ receives the
 total utility at most~$q-1$; a~contradiction.
\end{proof}

Finally, devising an ILP formulation, we show that \TIElong{\egal}{\opt} becomes
fixed-para\-meter tractable when parameterized by the combined parameter ``number
of manipulators and number of different utility values.'' This parameter covers
situations with few manipulators that have simple utility functions; in
particular, when few voters have $0/1$ utility functions. Together
with~\autoref{thm:egalTB} and~\autoref{thm:egalTBcomb},
following~\autoref{thm:egalTBilp} shows that neither few manipulators nor few
utility functions make \TIE{\egal}{\opt} fixed-parameter tractable, but only
combining these two parameters allows us to deal with the problem in \fpt{}
time.

\begin{theorem}\label{thm:egalTBilp}
 Let \udiff{} denote the number of different utility values
 and~$\manipnumber$~denote the number of manipulators. Then, parameterized
 by~$\manipnumber+\udiff$, \TIElong{\egal}{\opt} is fixed-parameter tractable.
\end{theorem}

\begin{proof}
 \newcommand{\typesset}{\ensuremath{\mathcal{T}}}
 \newcommand{\typevect}{\ensuremath{t}}
 \newcommand{\typecandidates}{\ensuremath{T}}
 \newcommand{\typesnr}{\ensuremath{|\typesset|}}
 We define the type of any candidate $c_i$ to be the size-\manipnumber{} vector
 $\typevect=(u_1(c_i), \allowbreak u_2(c_i), \ldots ,u_\manipnumber(c_i))$. Let
 \namedorderedsetof{\typesset}{\typevect}{\typesnr} be the set of all possible
 types. Naturally, the size of \typesset{} is upper-bounded by
 $\udiff^\manipnumber$.
 We denote the set of candidates of type~$\typevect_i \in \typesset$ by $\typecandidates_i$.
 Now, the ILP formulation of the
 problem using exactly $\typesnr+1$ variables reads as follows. For each type
 $\typevect_i \in \typesnr$, we introduce variable $x_i$ indicating the number
 of candidates of type $\typevect_i$ in an optimal \kelgroup{}. We use variable
 $s$ to represent the minimal value of the total utility achieved by
 manipulators. We define the following ILP with the goal to maximize $s$
  (indicating the utility gained by the least satisfied manipulator) subject to:
 \begin{align}
  \forall{\typevect_i \in \typesset} \colon& x_i \le |\typecandidates_i|
   \label{for:avail},\\
  &\sum_{\typevect_i \in \typesset} x_i = k
   \label{for:corr-kelgroup},\\
   \forall{\ell \in [\manipnumber]} \colon& \sum_{\typevect_i \in \typesset} x_i
   \cdot \typevect_i[\ell] \geq s.
   \label{for:util}
 \end{align}
 Constraint set \eqref{for:avail} ensures that the solution is achievable with
 given candidates. Constraint \eqref{for:corr-kelgroup} guarantees a choice of
 an \elgroup{} of size $k$. The last set of constraints imposes that $s$ holds
 at most the minimal value of the total utility gained by manipulators. By a
 famous result of \citet{Len83}, this ILP formulation with the number of
 variables bounded by $\udiff^\manipnumber+1$ yields that \TIElong{\egal}{\opt}
 is fixed-parameter tractable when parameterized by the combined parameter
 $\manipnumber + \udiff$.
\end{proof}

\section{Complexity of Coalitional Manipulation}
\label{sec:coal-man}
In the previous section, we have seen that breaking ties optimistically or
pessimistically\textemdash{}an essential subtask to be solved by the
manipulators in general\textemdash{}can be computationally challenging; in most
cases, however, this problem turned out to be computationally easy. In this
section, we move on to our full framework and analyze the computational
difficulty of voting strategically for a coalition of manipulators. To this end,
we formalize our central computational problem. Let~\calR be a multiwinner
voting rule and let~\calF be a multiwinner tie-breaking rule.

\probDefLong[$\eval \in \{\util, \egal, \candegal\}$]
{\CMlong{\calR}{\calF}{\eval}}{\CM{\calR}{\calF}{\eval}}
{An election $(C,V)$, a searched \elgroup{} size~$k<|C|$, \manipnumber{} manipulators
represented by their utility functions \namedorderedsetof{U}{u}{\manipnumber}
such that $\forall_{i \in [r]} \; u_i\colon C \rightarrow \naturals$, and a
non-negative, integral evaluation threshold~$q$.}
{Is there a size-\manipnumber{} multiset~$W$ of manipulative votes over~$C$ such
that \kelgroup{} $S \subset C$ wins the election $(C,V \cup W)$ under \calR and
\calF, and $\eval(S) \ge q$?}

The \CM{\calR}{\calF}{\eval} problem is defined very generally; namely, one can
consider any multiwinner voting rule \calR{} (in particular, any single-winner
voting rule is a multiwinner voting rule with~$k=1$). In our paper, however, we
focus on \ellBloc{}; hence, from now on, we narrow down our analysis of
\CM{\calR}{\calF}{\eval} to the \CM{\ellBloc{}}{\calF}{\eval} problem.

In line with our intention to model optimistic and pessimistic attitudes of
manipulators, we require that the evaluation of an optimistic/pessimistic
tie-breaking rule~\calF{} matches the manipulator's evaluation. 
More formally, for
every~$\eval \in \{\util, \egal, \candegal\}$, we focus on variants
of~$\CM{\ellBloc}{$\calF$}{$\eval$}$ where~$\calF \in \{\Flex,
\calF^{\eval}_{\opt}, \calF^{\eval}_{\pess}\}$.\footnote{The excluded problem
variants might indeed be relevant for situations where a tie-breaking goal is
different to manipulators' goal while utility values represent a somehow
``objective'' measure (e.g.,\,utility values represent monetary costs and,
although manipulators evaluate \elgroup{}s in the egalitarian way,
tie-breaking goal is to minimize the total cost); however such cases are beyond
the scope of this work.} We always allow lexicographic tie-breaking because it
models cases where a tie-breaking rule is fixed, known to all voters and, more
importantly, irrelevant of manipulators' utility functions.

On the way to show our results, we also use a restricted version of
\CMlong{\ellBloc{}}{\calF}{\eval} that we call
\emph{\CCMlong{\ellBloc{}}{\calF}{\eval}}. In this variant, the input stays the
same, but all manipulators cast exactly the same vote to achieve the objective.

To increase readability, we decided to represent manipulators by their utility
functions. As a consequence, we frequently use, for example, $u_1$ referring to
the manipulator itself, even if we do not care about the values of utility
function $u_1$ at the moment of usage. In the paper, we also stick to the term
``voters'' meaning the set $V$ of voters of an input election. We never call
manipulators ``voters''; however, we speak about the manipulative votes they
cast.

As for the encoding of the input of \CM{\calR}{\calF}{\eval}, we use a standard
assumption; namely, that the number of candidates, the number of voters, and the
number of manipulators are polynomially upper-bounded in the size of the input.
Analogously to \TIElong{\eval}{\bhav}, both the evaluation threshold and the
utility function values are encoded in binary.

\subsection{Utilitarian \& Candidate-Wise Egalitarian: Manipulation is
Tractable}\label{sec:utilManip} We show that \CMlong{\ellBloc}{\calF}{\eval} can
be solved in polynomial time for any constant~$\ell\in\naturals$, any $\eval \in
\{\util,\candegal\}$, and any $\calF{} \in
\{\Flex,\Fopt^{\eval},\Fpess^{\eval}\}$. Whereas in general, for~$n$~being the
input size, our algorithm requires~$O(n^5)$ steps, for~\Bloc{} (i.e.,\
$\ell=k$), we give a better, quadratic-time algorithm (with respect to~$n$).

In several proofs in \autoref{sec:utilManip} we use the value of a candidate for
manipulators (coalition) and say that a candidate is more valuable or less
valuable than another candidate. Although we cannot directly measure the value
of a candidate for the whole manipulators' coalition in general, thanks to
\autoref{obs:util-functions}, we can assume a single non-zero utility function
when discussing the utilitarian and candidate-wise egalitarian variants. Thus,
assigning a single value to each candidate is justified.

We start with an algorithm solving the general case of
\CMlong{\text{\ellBloc{}}}{\calF}{\eval}, $\eval \in \{\util,\candegal\}$,
$\calF{} \in \{\Flex, \Fopt^{\eval}, \Fpess^{\eval}\}$. The basic idea is to
``guess'' the lowest final score of a member of a~\kelgroup{} and (assuming some
lexicographic order over the candidates) the least preferred candidate of
the~\kelgroup{} that obtains the lowest final score. Then, the algorithm finds
an optimal manipulation leading to a~\kelgroup{} represented by the guessed
pair. Since there are at most polynomially-many (with respect to the input size)
pairs to be guessed, the described recipe gives a polynomial-time algorithm.

\begin{theorem}\label{thm:genCMinP}
 Let $m$ denote the number of candidates, $n$ the number of voters, $k$ the size
 of a searched \elgroup{}, and \manipnumber{} the number of manipulators. One
 can solve \CMlong{\text{\ellBloc{}}}{\calF}{\eval} in
 $O(k^2m^2(n+\manipnumber))$ time for any $\eval \in \{\util,\candegal\}$ and
 $\calF \in \{\Flex,\Fopt^{\eval},\Fpess^{\eval}\}$.
\end{theorem}
\begin{proof}
 We prove the theorem for the lexicographic tie-breaking rule \Flex{}. This is
 sufficient since, using~\autoref{prop:lex-simulation-util}, one can generalize
 the proof for the cases of utilitarian and candidate-wise egalitarian variants.
 The basic idea of our algorithm is to fix certain parameters of a solution and
 then to reduce the resulting subproblem to a variant of the \textsc{Knapsack}
 problem with polynomial-sized weights. The algorithm iterates through all
 possible value combinations of the following two parameters:
  \begin{itemize}
   \item
       the lowest final score~$z < |V \cup W|$ of any member of~the~\kelgroup{}
       and
   \item 
       the candidate~$\hat{c}$ that is the least preferred member of the
       \kelgroup{} with final score~$z$ with respect to tie-breaking rule~\Flex.
  \end{itemize}

 Having fixed~$z$ and $\hat{c}$, let $C^{+}$~denote the set of candidates who
 get at least $z+1$ approvals from the non-manipulative voters or who are
 preferred to~$\hat{c}$ according to~\Flex{} and get exactly $z$~approvals from
 the non-manipulative voters. Assuming that the combination of parameter values
 is correct, all candidates from $C^{+} \cup \{\hat{c}\}$ must belong to the
 \kelgroup{}. Let $k^{+}:=|C^{+}|$. For sanity, we check whether $k^{+} < k$,
 that is, whether candidate~$\hat{c}$ can belong to the \kelgroup{} if the
 candidate obtains final score~$z$. We discard the corresponding combination of
 solution parameter values if the check fails. Next, we ensure that $\hat{c}$
 obtains the final score exactly~$z$. If $\hat{c}$~receives less
 than~$z-\manipnumber$ or more than~$z$ approvals from non-manipulative voters,
 then we discard this combination of solution parameter values. Otherwise,
 let~$\hat{s}:=z - \score_V(\hat{c})$ denote number of additional approvals
 candidate~$\hat{c}$ needs in order to get final score~$z$. Let
 $k^*:=k-k^{+}-1$~be the number of remaining (not yet fixed) members of the
 \kelgroup{}. Let $s^*:=\manipnumber\cdot \ell - \hat{s}$ be the number of
 approvals to be distributed to candidates in $C \setminus \{\hat{c}\}$.
 
 Now, the manipulators have to influence further $k^*$~candidates to join the
 \kelgroup{} (so far only consisting of~$C^{+}\cup\{\hat{c}\}$) and distribute
 exactly $s^*$~approvals in total to candidates in $C \setminus \{\hat{c}\}$ but
 at most \manipnumber~approvals per candidate. To this end, let $C^*$~denote the
 set of candidates which can possibly join the \kelgroup{}. For each
 candidate~$c \in C \setminus (C^{+} \cup \{\hat{c}\})$ it holds that $c \in
 C^*$ if and only if
 \begin{enumerate}
   \item $z-\manipnumber \le \score_V(c) \le z-1$ if $c$~is preferred to
    $\hat{c}$ with respect to \Flex, or
   \item $z-\manipnumber+1 \le \score_V(c) \le z$ if $\hat{c}$~is preferred to
    $c$ with respect to \Flex.
  \end{enumerate}
 A straightforward idea is to select the $k^*$~elements from~$C^*$ which have
 the highest values (that is, utility) for the coalition. However, there can be
 two issues: First, $s^*$ might be too small; that is, there are too few
 approvals to ensure that each of the $k^*$~best-valued candidates gets the
 final score at least~$z$ (resp.\ at least $z+1$). Second, $s^*$ might be too
 large; that is, there are too many approvals to be distributed so that there is
 no way to do this without causing unwanted candidates to get a final score of
 at least~$z$ (resp.\ at least $z+1$).
 
 Fortunately, we can easily detect these cases and deal with them efficiently.
 In the former scenario we reduce the remaining problem to an instance of
 \textsc{Exact $k$-item Knapsack}---the problem in which, for a given set of
 items, their values and weights, and a knapsack capacity, we search for $k$
 items that maximize the overall value and do not exceed the knapsack capacity.
 In the latter case, we show that we can discard the corresponding combination
 of solution parameters.

 First, if $s^*\le \manipnumber\cdot k^*$, then one can certainly distribute all
 $s^*$~approvals (e.g.,\ to the $k^*$~candidates that will finally join the
 \kelgroup{}). Of course, it could still be the case that there are too few
 approvals available to push the desired candidates into the \kelgroup{} in a
 greedy manner. To solve this problem, we build an \textsc{Exact $k$-item
 Knapsack} instance where each candidate in $C^*$ is mapped to an item. We set
 the weight of each~$c^*\in C^*$ to $z-\score_V(c^*)$ if $c^*$~is preferred to
 $\hat{c}$ with respect to \Flex{} and otherwise to $(z+1)-\score_V(c^*)$. We set
 the value of each~$c^* \in C^*$ to be equal to the utility that candidate $c^*$
 contributes to
 the manipulators. Now, an optimal solution (given the combinations of
 parameter values is correct) must select exactly~$k^*$ elements from~$C^*$ such
 that the total weight is at most~$s^*$. This corresponds to \textsc{Exact
 $k$-item Knapsack} if we set our knapsack capacity to~$s^*$. Furthermore,
 finding any such set with maximum total value leads to an optimal solution.
 Even if the final total weight $s'$ of the chosen elements is smaller than
 $s^*$, we can transfer the \textsc{Exact $k$-item Knapsack} solution to the
 correct solution of our problem. The total weight corresponds to the number of
 approvals used. Thus, with the \textsc{Exact $k$-item Knapsack} solution we
 spend $s'$ approvals and, because of the monotonicity of \ellBloc{} together
 with the assumption that $s^*\le \manipnumber\cdot k^*$, we use $s^*-s'$
 approvals to approve the chosen candidates even more.

 Second, if $s^*>\manipnumber\cdot k^*$, then one can certainly ensure for any
 set of $k^*$~candidates from~$C^*$ the final score at least~$z$ (resp.\ at least
 $z+1$). In many cases, it will not be a problem to distribute the approvals;
 for example, one can safely spend up to \manipnumber~approvals for each
 candidate from~$C\setminus C^*$, that is, to candidates that have no chance to
 get enough points to join the \kelgroup{} or to candidates which are already
 fixed to be in the \kelgroup{}. Furthermore, each candidate from~$C^*$ can be
 safely approved $z-\score_V(c^*)-1$~times (resp.\ $z-\score_V(c^*)$~times)
 without reaching final score~$z$ (resp.\ $z+1$). We denote by~$s^{+}$ the total
 number of approvals which can be safely distributed to candidates
 in~$C\setminus\{\hat{c}\}$ without causing one of the candidates from~$C^*$ to
 reach score at least~$z$ (resp.\ at least $z+1$). If $s^*\le s^{+}+
 \manipnumber \cdot k^*$ (note that we also assume $s^* > \manipnumber \cdot
 k^*$), then we can greedily push the $k^*$~most-valued candidates from~$C^*$
 into the \kelgroup{} (spending $\manipnumber \cdot k^*$~approvals) and then
 safely distribute the remaining approvals within~$C\setminus\{\hat{c}\}$ as
 discussed. If $s^*>s^{+}+ \manipnumber \cdot k^*$, then there is no possibility
 of distributing approvals in a way that $\hat{c}$ is part of the \kelgroup{}.
 Towards a contradiction let us assume that $\hat{c}$ is part of the
 \kelgroup{} obtained after distributing $s^{+}+ \manipnumber \cdot k^* + 1$
 approvals. This means that we spend all possible $s^{+}$~approvals so that
 $\hat{c}$~is not beaten and $\manipnumber \cdot k^{*}$ approvals to
 push~$k^{*}$~candidates to the winning \kelgroup{}. Giving one more approval to
 some candidate~$c'$ from~$C^*$ that is not yet in the \kelgroup{}, by
 definition of~$C^*$ and~$s^+$, means that the score of~$c'$ is enough to push
 $\hat{c}$ out of the final \kelgroup{}; a contradiction. Consequently, for the
 case of $s^* > s^+ + \manipnumber \cdot k^*$, we discard the corresponding
 combination of solution parameters.

 As for the running time, the first step is sorting the candidates according to
 their values in $O(m(r+\log(m)))$ time. Then let us consider the running time
 of two cases $s^*\leq \manipnumber \cdot k^*$ and $s^* > \manipnumber \cdot
 k^*$ separately. In
 the former case, we solve \textsc{Exact $k$-item Knapsack} in $O(k^2mr)$ time by
 using dynamic programming based on analyzing all possible total weights of the
 selected items until the final value is reached~\citep[Chapter
 9.7.3]{KPP04}\footnote{In fact, \citet{KPP04} use dynamic programming based on
 all possible total values of~items. However, one can exchange all possible
 total values of items with all possible total weights of items and thus obtain
 an algorithm with running time polynomial in the maximum weight of items.} (note
 that the maximum possible total weight is upper-bounded by $k\manipnumber$). If
 $s^* > \manipnumber \cdot k^*$, then we approve at most $m$~candidates which
 gives the running time~$O(m)$. Thus, we can conclude that the running time of
 the discussed cases is $O(k^2mr)$. Additionally, there are at most
 $n+\manipnumber$ values of~$z$ and at most~$m$ values of~$\hat{c}$. Summarizing,
 we get the running time $O(k^2m^2(n+r))$.
\end{proof}

Next, we show that, actually, \CM{\Bloc}{\calF}{\eval} (i.e.,\ a special case
of~\CM{\ellBloc}{\calF}{\eval} where $\ell=k$) can be solved in quadratic-time,
that is, in practice, much faster than the general variant of the problem. On
our way to present this results, we first give an algorithm
for~\CCM{\ellBloc}{\calF}{\eval}. Then, we argue that it also
solves~\CM{\Bloc}{\calF}{\eval}. The algorithm ``guesses'' the minimum score
among all members of the winning \elgroup{} and then carefully (with respect to
the tie-breaking method) selects the best candidates that can reach this score.

\begin{proposition} \label{prop:utilConsistentPoly}
 Let $m$ denote the number of candidates, $n$ denote the number of voters, and
 $r$ denote the number of manipulators. Then one can solve
 \CCMlong{\ellBloc}{\calF}{\eval} in $O(m(m + r +n))$ time for any $\eval \in
 \{\util,\candegal\}$ and $\calF{} \in \{\Flex{}, \Fopt^{\eval},
 \Fpess^{\eval}\}$.
\end{proposition}
\begin{proof}

Consider an instance of \CCM{\ellBloc}{\Flex}{\eval} with an election $E=(C,V)$
where $C$ is a candidate set and $V$ is a multiset of non-manipulative votes,
\manipnumber~manipulators, an \elgroup{} size $k$, and a lexicographic order
$\lexpref$ used by \Flex{} to break ties. In essence, we introduce a constrained solution
form called a \emph{canonical solution} and argue that it is sufficient to
analyze only this type of solutions. Then we provide an algorithm that
efficiently seeks for an optimal canonical solution.

At the beginning, we observe that when manipulators vote consistently, then we
can arrange the top $\ell$ candidates of a manipulative vote in any order.
Hence, the solution to our problem is a size-$\ell$ subset (instead of an order)
of candidates which we call a set of \emph{supported candidates}; we call each
member of this set a \emph{supported candidate}.

\mypar{Strength order of the candidates.}
Additionally, we introduce a new order~\strpref{} of the candidates. It sorts
them descendingly with respect to the score they receive from voters and, as a
second criterion, according to the position~in~$\lexpref$.
Intuitively, the easier it is for some candidate to be a part of a winning
\kelgroup{}, the higher is the candidate's position in~$\strpref{}$. As a
consequence, we state~\autoref{obs:strpref}.
\begin{claim}\label{obs:strpref}
 Let us fix an instance of \CCM{\ellBloc}{\Flex}{\eval} and a solution $X$ which
 leads to a winning \kelgroup{} $S$. For every supported
 (resp.\ unsupported) candidate~$c$, the
 following holds:
 \begin{enumerate}
  \item If $c$ is part of the winning \kelgroup{}, then every supported (resp.\
   unsupported) predecessor of~$c$, according to $\strpref{}$, belongs to~$S$.
   \label{obs:strpref_part}
  \item If $c$ is not part of the winning \kelgroup{}, then every supported
   (resp.\ unsupported) successor of~$c$, according to $\strpref{}$, does not
   belong to $S$. \label{obs:strpref_notpart}
 \end{enumerate}
\end{claim}
\begin{proof}
 Fix an instance of~\CCM{\ellBloc}{\Flex}{\eval}, a solution~$X$, and a
 winning~\kelgroup{}~$S$. Let us consider the respective order~\strpref{} over
 the candidates in the instance.

 We first show that statement~\ref{obs:strpref_part} regarding supported
 candidates holds. According to the statement, fix some supported candidate~$c
 \in S$ and let~$p$ be a predecessor of~$c$ (according to~\strpref{}). Towards a
 contradiction, let us assume that~$p \notin S$. This implies that either (i)
 the score of~$p$ is smaller than the score of~$c$ or (ii) their scores are the
 same but~$c \lexpref{} p$. Let us focus on case (i). Both considered
 candidates are supported by all manipulators (note that manipulators vote
 consistently). Thus, as a consequence of~$p \strpref{} c$, we have that the
 score of~$p$ is at least as high as the score of~$c$; a contradiction. Next,
 consider case~(ii), where $p$ and~$c$ have the same scores. Consequently, the
 mutual order of~$c$ and $p$ in~\strpref{} is the same as their order
 in~$\lexpref$ (in other words, the order of~$c$ and~$p$ in~\strpref{} does not
 depend on scores of~$c$ and~$p$ because those must be the same prior to any
 manipulation). Since~$c \lexpref p$, it follows that, by definition
 of~\strpref{}, it must hold that~$c \strpref{} p$; a contradiction again.
 Eventually, we obtain that~$p$ has to be part of~$S$ which completes the
 argument.

 An analogous approach leads to proofs for the remaining three cases stated in
 the theorem.%
\end{proof}

\autoref{obs:strpref} justifies thinking about \strpref{} as a ``strength
order''; hence, in the proof we use the terms \emph{stronger} and \emph{weaker}
candidate. Using \autoref{obs:strpref}, we can fix some candidate $c$ as the
weakest in the winning \kelgroup{} and then infer candidates that have to be and
that cannot be part of this \kelgroup{}. To formalize this idea, we introduce
the concept of a \emph{canonical solution}.

\mypar{Canonical solutions.}
Assuming the case where $k \leq \ell$, we call a solution $X$ leading to a
winning \kelgroup{}~$S$ canonical if all candidates of the winning \elgroup{}
are supported; that is, $S \subseteq X$. In the opposite case, $k > \ell$,
solution $X$ is canonical if $X \subset S$ and $X$ is a set of the $\ell$
weakest candidates in $S$. For the latter case, the formulation describes the
solution which favors supporting weaker candidates first and ensures that no
approval is given to a candidate outside the winning \kelgroup{}. 

Canonical solutions are achievable from every solution without changing the
outcome. Observe that one cannot prevent a candidate from winning by supporting
the candidate more because this only increases the candidate's score.
Consequently, we can always transfer approvals to all candidates from the
winning \kelgroup{}. For the case of $k > \ell$, we then have to rearrange the
approvals in such a way that only the weakest members of the \kelgroup{} are
supported. However, such a rearrangement cannot change the outcome because,
according to \autoref{obs:strpref}, we can transfer an approval from some
stronger candidate $c$ to weaker $c'$ keeping both of them in the winning
\kelgroup{}.

\mypar{Dropped and kept candidates.}
Observe that for every solution (including canonical solutions), we can always
find the strongest candidate who is not part of the winning \elgroup{}. We call
this candidate the \emph{dropped candidate}. Note that we use the strength order
in the definition of the dropped candidate; this order does not take
manipulative votes into account. Moreover, without loss of generality, we can
assume that the dropped candidate is not a supported candidate. This is because
if the dropped candidate is not in the winning \kelgroup{} even if supported,
then we can support any other candidate outside of the winning \kelgroup{}
without changing the winning \kelgroup{} (see \autoref{obs:strpref}). There
always exists some candidate to whom we can transfer our support because $\ell <
m$.  Naturally, by definition of the dropped candidate, all candidates stronger
than the dropped candidate are members of the winning \kelgroup{}. We call these
candidates \emph{kept candidates}.

\mypar{High-level description of the algorithm.}
The algorithm solving \CCM{\ellBloc}{\Flex}{\eval} iteratively looks for an
optimal canonical solution for every possible (non-negative) number $t$ of kept
candidates (alternatively the algorithm checks all feasible possibilities of
choosing the dropped candidate). Then, the algorithm compares all solutions and
picks one that is resulting in an \elgroup{} liked the most by the manipulators.
Observe that $k-\ell \leq t \leq k$. The upper bound~$k$ is the consequence of
the fact that each kept candidate is (by definition) in the winning \kelgroup{}.
Since all candidates except for kept candidates have to be supported to be part
of the winning \elgroup, we need at least $k-\ell$ kept candidates, in order to
be able to complete the \kelgroup.

\mypar{Running time.}
To analyze the running time of the algorithm described in the previous
paragraph, several steps need to be considered. At the beginning we have to
compute values of candidates and then sort the candidates with respect to their
value. This step runs in $O(\manipnumber{}m + m \log{}m)$ time. Similarly,
computing~$\strpref{}$ takes $O(\ell{}n + m \log{}m)$ time. Having both
orders, Procedure~\ref{alg} (described in detail later in this proof) needs
$O(m)$ to find an optimal canonical solution for some fixed number $t$ of kept
candidates. Finally, we have at most $\ell+1$ possible values of $t$. Summing
the times up, together with the fact that $\ell<m$, we obtain a running time
$O(m(m + r +n))$.

\mypar{What remains to be done.}
Procedure~\ref{alg} describes how to look for an optimal canonical solution for
a fixed number $t$ of kept candidates. First, partition the candidate set
in the following way. By $\secured$ we denote the kept candidates (which are the
top $t$ candidates according to \strpref{}). Consequently, the $(t+1)$-st
strongest candidate is the dropped candidate; say $c^*$. For every value of $t$,
the corresponding dropped candidate, by definition, is not allowed to be part of
the winning \elgroup{}. Let
\begin{align*}
D=\{\secured \cup \{c^*\} \not\owns c \mid &(\score_V(c)+ \manipnumber >
\score_V(c^*)) \vee \\
&(\score_V(c) + \manipnumber = \score_V(c^*) \wedge
c\lexpref{} c^*)\}
\end{align*}
be the set of \emph{distinguished candidates}. Each distinguished candidate, if
supported, is preferred over $c^*$ to be selected into the winning \kelgroup{}.
Consequently, the distinguished candidates are all candidates who can
potentially be part of the winning \kelgroup{}. We remark that to fulfill our
assumption that the dropped candidate is not part of a winning \elgroup{}, it is
obligatory to support at least $k-t$ distinguished candidates. Note that
$\secured \cup\: \{c^*\} \cup D$ is not necessarily equal to $C$. The
remaining candidates cannot be part of the winning \kelgroup{} under any
circumstances assuming $t$ kept candidates. Also, set~$D$ might consist of less
than~$k-t$ required candidates (which is the case when there are too few
candidates that, after supported, would outperform~$c^*$). If such a situation
emerges, we skip the respective value of~$t$. Making use of the described
division into $c^*$, $D$, and $\secured$, Procedure~\ref{alg} incrementally
builds set~$X$ of supported candidates associated with an optimal solution until
all possible approvals are used. Observe, that since $k < |C|$ and $\ell < |C|$,
it is guaranteed that for $t=k$ Procedure~\ref{alg} will return a feasible
solution for~$t$; in fact, this solution will always result in a winning
\elgroup{} consisting of all~$t$ kept candidates (irrespective of~$D$).

\begin{algorithm}[t!]
 \DontPrintSemicolon
 \SetKw {Or}{or}
 \caption{A procedure of finding an optimal set of supported candidates.
   \label{alg}}
 \SetKwInput{Input}{Input}\SetKwInOut{Output}{Output}
 \Input{Election $E=(C,V)$; number $\ell$ of approvals in \ellBloc{} rule; size
 $k$ of the winning \kelgroup{}; a partition of $C$ into kept candidates
 $\secured$ (such that $|\secured|=t$ and $k-\ell \leq t \leq k$), a dropped
 candidate $c^*$, and distinguished candidates $D$ (such that $|D|\geq k-t$).}
 \Output{Optimal supported candidates set $X$}

 $X \longleftarrow$ \{the $k-t$ most valuable candidates from $D\}$\; \label{alg:begin-firstX}
$X \longleftarrow X \cup$ \{$\min\{t,\ell-|X|\}$ arbitrary candidates from
$\secured\}$\; \label{alg:end-firstX}

 \If{$\ell\neq|X|$}{
  $A \longleftarrow \{$the $\ell-|X|$ weakest candidates from $C \setminus X\}$\;
  \label{alg:begin-sim}
  $B \longleftarrow \{$top $k$ strongest candidates from $X \cup A\}$\;
 $p \longleftarrow |B \setminus X|$\; \label{alg:end-sim}
 $X \longleftarrow X \cup \{$the $\ell-|X|-p$ weakest candidates from
 $C \setminus X\}$\; \label{alg:start_finalX}
 $X \longleftarrow X \cup \{$the $p$ most valuable candidates from $D \setminus
 X\}$\; \label{alg:end_finalX}
  }
 \Return{$X$}\;
\end{algorithm}

\mypar{Detailed description of the algorithm.}
Before studying Procedure~\ref{alg} in detail, consider~\autoref{fig:alg}
illustrating the procedure on example data. In line~\ref{alg:begin-firstX}, the
procedure builds set~$X$ of supported candidates using the $k-t$ best valued
distinguished candidates. Since only the distinguished candidates might be a
part of the winning \kelgroup{} besides the kept candidates, there is no better
outcome achievable. Then, in line~\ref{alg:end-firstX}, the remaining approvals,
if they exist, are used to support kept candidates. This operation does not
change the resulting \kelgroup{}. Then Procedure~\ref{alg} checks whether all
$\ell$ approvals were used; that is, whether $\ell=|X|$. If not, then there are
exactly $\ell-|X|$ remaining approvals to use. Note that at this stage set $X$
contains $k$ supported candidates who correspond to the best possible
\kelgroup{}, however, without spending all approvals. Let us call this
\kelgroup~$S$. It is possible that there is no way to spend the remaining
$\ell-|X|$ approvals without changing the winning \kelgroup{} $S$. Then
substitutions of candidates occur.  The new candidates in the \kelgroup{} can be
only those that are distinguished and so far unsupported whereas the exchanged
ones can be only so far supported distinguished candidates. This means that each
substitution lowers the overall value of the winning \kelgroup{}. So, the best
what can be achieved is to find the minimal number of substitutions and then
pick the most valuable remaining candidates from $D$ to be substituted. The
minimal number of substitutions can be found by analyzing how many candidates
would be exchanged in the winning \kelgroup{} if the weakest $\ell-|X|$
previously unsupported candidates were supported. The procedure makes such a
simulation and computes the number~$p$ of necessary substitutions, in lines
\ref{alg:begin-sim}-\ref{alg:end-sim}.  Supporting the $\ell-|X|-p$ weakest
unsupported candidates and then the $p$ most valuable so far unsupported
distinguished candidates gives the optimal \kelgroup{} for~$t$~kept candidates
(when all approvals are spent). Note that the number $\ell$ of approvals is
strictly lower than the number of candidates, so one always avoids supporting
$c^*$.

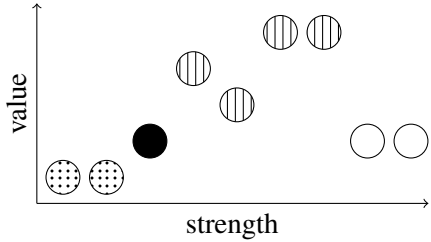
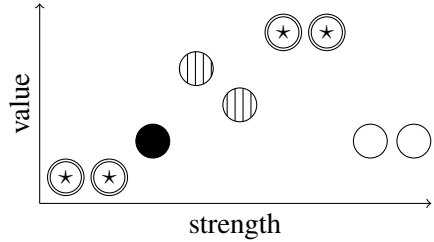
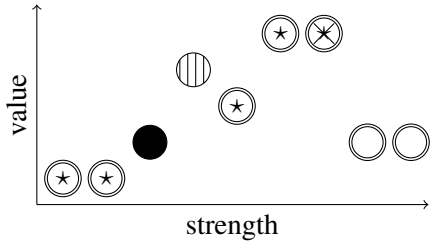
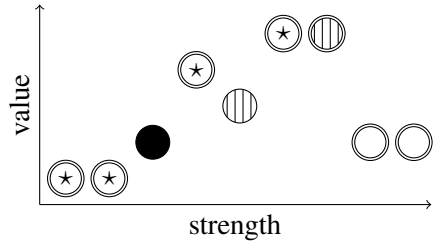
\begin{figure}
 \begin{minipage}[t]{0.48\textwidth}
  \begin{tikzpicture}[scale=\textwidth/10.5cm]
 \pgfdeclarepatternformonly{empty}{\pgfpointorigin}{\pgfpoint{1}{1}}
 {\pgfpoint{1}{1}} {}
 \coordinate (point0) at (0,-0.6);
 \foreach[count=\i] \u/\shape/\pat/\desc in {%
  0/circle/dots/,
  0/circle/dots/,
  1/circle//,
  3/circle/vertical lines/,
  2/circle/vertical lines/,
  4/circle/vertical lines/,
  4/circle/vertical lines/,
  1/circle/empty/,
  1/circle/empty/%
 }
 \node (c\i) at (\i-0.4,\u/1.2) [
   pattern = \pat, \shape, draw,
   minimum size=4.5mm, inner sep=0mm] {\desc};
  \draw [->] (point0) to node [below] {strength}
  (9,-0.6);
  \draw [->] (point0) to node [sloped, above] {value} (0,4);
\end{tikzpicture}
\subcaption{The division of the candidates into the kept candidates (dotted),
 the dropped candidate (filled), the distinguished candidates (vertical lines),
 and the others who cannot be a part of the winning
 \elgroup{}\label{fig:cand-division}.
}
\end{minipage}
\hfill
\begin{minipage}[t]{0.48\textwidth}
  \begin{tikzpicture}[scale=\textwidth/10.5cm]
 \coordinate (point0) at (0,-0.6);
 \foreach[count=\i] \u/\shape/\pat/\desc/\d in {%
  0/circle/empty/$\star$/double,
  0/circle/empty/$\star$/double,
  1/circle///,
  3/circle/vertical lines//,
  2/circle/vertical lines//,
  4/circle/empty/$\star$/double,
  4/circle/empty/$\star$/double,
  1/circle/empty//,
  1/circle/empty//%
 }
 \node (c\i) at (\i-0.4,\u/1.2) [
   pattern = \pat, \shape, draw, \d,
   minimum size=4.5mm, inner sep=0mm] {\desc};
  \draw [->] (point0) to node [below] {strength}
  (9,-0.6);
  \draw [->] (point0) to node [sloped, above] {value} (0,4);
\end{tikzpicture}
\subcaption{An illustration of lines \ref{alg:begin-firstX}-\ref{alg:end-firstX}
of Procedure~\ref{alg}. The double-edged candidates form set $X$. The
starred candidates would form the winning \kelgroup{} if the double-edged
candidates were supported.\label{fig:opt-kegroup}
}
\end{minipage}
\begin{minipage}[t]{0.48\textwidth}
  \begin{tikzpicture}[scale=\textwidth/10.5cm]
 \coordinate (point0) at (0,-0.6);
 \foreach[count=\i] \u/\shape/\pat/\desc/\d in {%
  0/circle/empty/$\star$/double,
  0/circle/empty/$\star$/double,
  1/circle///,
  3/circle/vertical lines//,
  2/circle/empty/$\star$/double,
  4/circle/empty/$\star$/double,
  4/circle/empty/$\star$/double,
  1/circle/empty//double,
  1/circle/empty//double%
 }
 \node (c\i) at (\i-0.4,\u/1.2) [
   pattern = \pat, \shape, draw, \d,
   minimum size=4.5mm, inner sep=0mm, text width=0.3cm, align=center] {\desc};
  \draw [->] (point0) to node [below] {strength}
  (9,-0.6);
  \draw [->] (point0) to node [sloped, above] {value} (0,4);
  \draw (c7.center) node[cross out, draw, minimum size=1mm]{};
\end{tikzpicture}
\subcaption{Supporting the weakest possible candidates to use all approvals.
 The winning \elgroup{} changes. The winners are marked with stars while the
 candidate who is no more in the winning \elgroup{} is crossed out. Such a
 simulation is done in lines \ref{alg:begin-sim}-\ref{alg:end-sim}. As a result,
 the minimum number of substitutions in the \kelgroup{} is computed.
 }
\end{minipage}
\hfill
\begin{minipage}[t]{0.48\textwidth}
  \begin{tikzpicture}[scale=\textwidth/10.5cm]
 \coordinate (point0) at (0,-0.6);
 \foreach[count=\i] \u/\shape/\pat/\desc/\d in {%
  0/circle/empty/$\star$/double,
  0/circle/empty/$\star$/double,
  1/circle///,
  3/circle/empty/$\star$/double,
  2/circle/vertical lines//,
  4/circle/empty/$\star$/double,
  4/circle/vertical lines//double,
  1/circle/empty//double,
  1/circle/empty//double%
 }
 \node (c\i) at (\i-0.4,\u/1.2) [
   pattern = \pat, \shape, draw, \d,
   minimum size=4.5mm, inner sep=0mm] {\desc};
  \draw [->] (point0) to node [below] {strength}
  (9,-0.6);
  \draw [->] (point0) to node [sloped, above] {value} (0,4);
\end{tikzpicture}
\subcaption{An illustration of the solution of the considered case computed by
 Procedure~\ref{alg} in lines~\ref{alg:start_finalX} to~\ref{alg:end_finalX}.
 One candidate from the \kelgroup{} presented in \autoref{fig:opt-kegroup} has
 to be substituted; naturally, it is optimal to pick the most valuable possible
 candidate as a replacement for the substituted one. Supported candidates are
 double-edged and the winning \kelgroup{} is starred.\label{fig:opt-solution}
}
\end{minipage}
\caption{An illustrative example of a run of Procedure~\ref{alg} for $t=2$,
nine candidates, \lBloc{7}, and \namedelgroup{4}. The horizontal position
indicates the strength of a candidate and the vertical position indicates the
value of a candidate. Since the number \manipnumber{} of manipulators
determines only the set of distinguished candidates, we do not specify
\manipnumber{} explicitly. We indicate the set of distinguished candidates
instead. Subfigures~\ref{fig:cand-division}~to~\ref{fig:opt-solution} step by
step present the execution of Procedure~\ref{alg} on the way to find an optimal
\namedelgroup{4}.\label{fig:alg}
}
\end{figure}
The algorithm we presented can be applied also for pessimistic and optimistic
evaluation because of the possibility of simulating these evaluations by a
lexicographic order in time $O(m(\manipnumber+\log(m)))$ (see
\autoref{prop:lex-simulation-util}).
\end{proof}

For \Bloc{}, we will show that manipulators can always vote identically to
achieve an optimal \kelgroup{}. In a nutshell, for every \elgroup{} the
manipulators can only increase the scores of its members by voting exactly for
them. This fact leads to the next corollary.
\begin{corollary}\label{cor:easyCM}
 Let $m$ denote the number of candidates, $n$ denote the number of voters, and
 $r$ denote the number of manipulators. One can solve
 \CMlong{\Bloc{}}{\calF}{\eval} in  $O(m(m + r +n))$ time for any $\eval \in
 \{\util,\candegal\}$ and $\calF{} \in \{\Flex, \Fopt^{\eval},
 \Fpess^{\eval}\}$.
\end{corollary}
\begin{proof}
 We show that for \CMlong{\Bloc{}}{\calF}{\eval} the manipulators have no
 incentive to deviate from one optimal profile (i.e., they vote in the same
 manner). Let us fix an optimal \kelgroup{} $S$. If there exists a candidate $c
 \in S$ which is not approved by some manipulator~$u^*$, then there exists also
 some candidate $c' \notin S$ which is approved by $u^*$ ($u^*$ approves at most
 ${k-1}$~candidates from $S$). Observe that in the \Bloc{} voting rule by shifting
 a candidate up in a preference order we only increase the candidate's score; as
 a result, we cannot prevent the candidate from winning by doing such a shift.
 Using this observation, we can exchange candidate~$c$ with candidate~$c'$ in
 the preference order of $u^*$ without preventing $c$ from winning. We repeat
 exchanging candidates until all manipulators approve only candidates from $S$.
 Then we obtain an optimal vote by fixing a preference order over those
 candidates arbitrarily (there might be more than one optimal vote but all of
 them place only candidates from set $S$ at the first $k$ places). Concluding,
 we can use the algorithm from \autoref{prop:utilConsistentPoly} which works in
 the given time.
\end{proof}

\subsection{Egalitarian: Hard Even for Simple Tie-Breaking}

In \autoref{sec:hardTB}, we showed that already breaking ties might be
computationally intractable. These intractability results only hold with respect
to the egalitarian evaluation and optimistic manipulators.
We now show that this intractability transfers to coalitional manipulation for
any tie-breaking rule and the egalitarian evaluation. This includes the
pessimistic egalitarian case which we consider to be highly relevant as it
naturally models searching for a ``safe'' voting strategy. 

\begin{proposition}\label{prop:tb-to-cm-reduction}
 For any tie-breaking rule~\calF{}, there is a polynomial-time many-one
 reduction from \TIElong{\egal}{\opt} to
 \CMlong{\text{\ellBloc{}}}{\calF}{\egal}.
\end{proposition}

\begin{proof}
We reduce an instance of \TIElong{\egal}{\opt} to
\CMlong{\text{$\ell$-Bloc}}{\calF}{\egal}; however, before we describe the
actual reduction, we present a useful observation concerning
\TIElong{\egal}{\opt} in the next paragraph.

Let us fix an instance $I$ of \TIElong{\egal}{\opt} with a confirmed set $C^+$, a
pending set $P$, a size $k$ of an \elgroup{}, a threshold $q$, and a set of
manipulators represented by a family $U$ of utility functions. We can construct
a new equivalent instance $I'$ of \TIElong{\egal}{\opt} with a
larger set of manipulators' utility functions $U' \supseteq U$. The construction
is a polynomial-time many-one reduction which proves that we can ``pump'' the
number of manipulators arbitrarily for instance $I$. To add a manipulator, it is
enough to set to~$q$ the utility that the manipulator gives to every candidate.
Naturally, such a manipulator cannot have the total utility smaller than~$q$, so
the correct solution for $I$ is also correct for $I'$. Contrarily, when there is
no solution for $I$, it means that for every possible \kelgroup{} $S'$ there is
some manipulator $\bar{u}$ such that $\egal_{\bar{u}}(S')<q$. Consequently, one
cannot find a solution for $I'$ as well, because the set of possible
\kelgroup{}s and their values of egalitarian utility do not change.

Now we can phrase our reduction from \TIElong{\egal}{\opt} to
\CMlong{\text{$\ell$-Bloc}}{\calF}{\egal}. Let us fix an instance $I$ of
\TIElong{\egal}{\opt} with a confirmed set $C^+$, a pending set $P$, a size $k$
of an \elgroup{}, a threshold $q$, and a set $U$ of \manipnumber{} utility
functions. Because of the observation about
``pumping'' instances of \TIElong{\egal}{\opt}, we can assume, without loss of
generality, that $\ell \cdot \manipnumber \geq k - |C^+|$ holds. In the
constructed instance of \CM{\text{$\ell$-Bloc}}{\calF}{\egal} equivalent to $I$,
we build an election that yields sets $P$ and $C^+$ and aim at an \elgroup{} of
size~$k$. However, it is likely that we need to add a set of dummy candidates
that we denote by~$D$. It is important to ensure that the dummy candidates
cannot be the winners of the constructed election. To do so, we keep the score
of each dummy candidate to be at most $1$, the score of each pending candidate
to be $\manipnumber+2$, and the score of each confirmed candidate to be at least
$2\manipnumber+3$. The construction starts from ensuring the scores of the
confirmed candidates. Observe, that in this step we add at most $(2\manipnumber
+ 3) \cdot |C^+|$ voters (in case $\ell=1$). If $\ell>|C^+|$, then we have to
add some dummy candidates in this step. We can upper-bound the number of the
added dummy candidates by $((2\manipnumber+3) \cdot |C^+|)(\ell-1)$ (this bound
is not tight). Analogously, we add new voters such that each pending candidate
has score exactly $\manipnumber+2$. At this step we have the election where we
are able to spend $\ell \cdot \manipnumber \geq k-|C^+|$ approvals. We can
select every possible subset of pending candidates to form the winning
\kelgroup{} by approving candidates in this subset exactly once. However, to be
sure that we are able to distribute all approvals such that there is no tie, we
ensure that the remaining $(\ell \cdot \manipnumber)-(k-|C^+|)$ approvals can be
distributed to some candidates without changing the outcome. To achieve this
goal we add exactly $(\ell \cdot \manipnumber)-(k-|C^+|)$ dummy candidates with
score~$0$. We set the evaluation threshold of the newly constructed instance
to~$q$.

By our construction, we are always able to approve enough pending candidates to
form a \kelgroup{} without considering ties, and we cannot make a dummy candidate
a winner under any circumstances. Thus, if \TIElong{\egal}{\opt} has a solution
$S$, then we approve every candidate $c \in S$ such that $c$ was in the pending
set $P$ before, and we obtain a solution to the reduced instance. In the
opposite case, if there is no such a \kelgroup{} whose egalitarian utility value
is at least $q$, then the corresponding instance of
\CMlong{\text{$\ell$-Bloc}}{\calF}{\egal} also has no solution since the
possible \kelgroup{}s are exactly the same. The reduction runs in polynomial
time.
\end{proof}
Observe that the reduction proving~\autoref{prop:tb-to-cm-reduction} does not
change the \elgroup{} size~$k$. Additionally, the increase of the number of
manipulators in resulting instances is polynomially bounded in the \elgroup{}
size~$k$ of input instances. This is due to the fact that even if we need to
``pump'' an initial instance to achieve~$\ell \cdot \manipnumber \geq k-|C^+|$,
then we add at most
$\left\lceil \frac{k-|C^+|}{\ell} \right\rceil \leq k$ manipulators.
Thus, together with \autoref{thm:egalTB} and \autoref{thm:egalTBcomb},
\autoref{prop:tb-to-cm-reduction} leads to the following corollary.
\begin{corollary}\label{cor:egalCMhard}
 Let~\calF{} be an arbitrary tie-breaking rule. Then,
 \CMlong{\text{\ellBloc{}}}{\calF}{\egal} is \np-hard. Let \manipnumber{} denote
 the number of manipulators, $q$ denote the evaluation threshold and $k$ denote
 the size of an \elgroup{}. Then, parameterized by $\manipnumber+k$,
 \CM{\text{\ellBloc{}}}{\calF}{\egal} is \wone-hard. Parameterized by $k$,
 \CM{\text{\ellBloc{}}}{\calF}{\egal} is \wtwo-hard even if $q=1$ and every
 manipulator only gives either utility one or zero to each candidate.
\end{corollary}

Finally, by using ideas from \autoref{thm:genCMinP} and an adaptation of the ILP
from \autoref{thm:egalTBilp} as a subroutine, we show that, for the combined
parameter ``the number of manipulators and the number of different utility
values'', fixed-parameter tractability of \TIElong{\egal}{\opt} transfers to
coalitional manipulation for both optimistic and pessimistic tie-breaking.
\begin{theorem}\label{egalCMilp}
  Let \manipnumber{} denote the number of manipulators and \udiff{} denote the
  number of different utility values. Parameterized
  by~$\manipnumber+\udiff$, \CMlong{\text{\ellBloc{}}}{\calF}{\egal} with
  $\calF \in \{\Fpess^{\egal}, \Fopt^{\egal{}}\}$ is fixed-parameter tractable.
\end{theorem}

\begin{proof}
 \newcommand{\tiednr}{\ensuremath{b}}
 \newcommand{\toconfirm}{\ensuremath{p}}
 \newcommand{\xfinmem}{\ensuremath{x_i^{j\bullet}}}
 \newcommand{\xvarname}{\ensuremath{x}}
 \newcommand{\xpending}{\ensuremath{\xvarname_i^{j}}}
 \newcommand{\xconfirmed}{\ensuremath{\xvarname_i^{j+}}}
 \newcommand{\worstmanip}{\ensuremath{r^*}}
 \newcommand{\worstutil}{\ensuremath{s^*}}
 \newcommand{\typesset}{\ensuremath{\mathcal{T}}}
 \newcommand{\typevect}{\ensuremath{t}}
 \newcommand{\typecandidates}{\ensuremath{T}}
 \newcommand{\typesnr}{\ensuremath{|\typesset|}}
 \newcommand{\lowestscore}{\ensuremath{z}}
 \newcommand{\manipnumberandzero}{{\ensuremath{[\manipnumber] \cup \{0\}}}}
 \newcommand{\reqapprov}{\ensuremath{o}}
 \newcommand{\notspendapprov}{\ensuremath{\bar{o}}}
 \newcommand{\elgroupsize}{\ensuremath{k}}
 \newcommand{\groupsset}{\ensuremath{\matchal{G}}}
 \newcommand{\groupvar}{\ensuremath{G}}
 \newcommand{\group}{\ensuremath{\groupvar_i^j}}
 \newcommand{\pprocedure}{\ensuremath{\mathcal{P}}}

 In a nutshell, we divide \CM{\text{\ellBloc{}}}{\Fpess^{\egal}}{\egal} and
 \CM{\text{\ellBloc{}}}{\Fopt^{\egal}}{\egal} into subproblems solvable in
 \fpt{} time with respect to the combined parameter ``number of manipulators and
 number of different utility values.'' We show that solving polynomially many
 subproblems is enough to solve the problems.

 \mypar{The main idea.}
 We split the proof into two parts.
 In the first part, we define subproblems and show how to find a solution
 assuming that the subproblems are solvable in \fpt{} time with respect to the
 parameter. In the second part, we show that, indeed, the subproblems are
 fixed-parameter tractable using their ILP formulations. The inputs for
 \CM{\text{\ellBloc{}}}{\Fpess^{\egal}}{\egal} and
 \CM{\text{\ellBloc{}}}{\Fopt^{\egal}}{\egal} are the same, so let us consider
 an arbitrary input with an election $E=(C, V)$ where $|V|=n$, $|C|=m$, a size
 \elgroupsize{} of an \elitegroup, and \manipnumber{} manipulators represented
 by a set \namedorderedsetof{U}{u}{\manipnumber} of their utility functions. Let
 \udiff{} be the number of different utility values.

 An election resulting from a manipulation and a corresponding \kelgroup{}
 emerging from the manipulation can be described by three non-negative integer
 parameters: 
 \begin{enumerate}
  \item the lowest final score $\lowestscore$ of any member of the
   \kelgroup{};
  \item the number \toconfirm{} of \emph{promoted candidates} from the
   \kelgroup{} with a score higher than~\lowestscore{} that, at the same time,
   have score at most \lowestscore{} without taking manipulative votes into
   consideration;
   \label{it:promoted}
  \item the number \tiednr{} of \emph{border candidates} with score
   \lowestscore{}\label{it:border}.
 \end{enumerate}
 Observe that if as a result of a manipulation the lowest final score of members
 in a final \kelgroup{} is $\lowestscore$, then the promoted candidates are part
 of the \kelgroup{} regardless of the tie-breaking method used. For border
 candidates, however, it might be necessary to run the tie-breaking rule to
 determine the \kelgroup{}. In other words, border candidates become pending
 candidates unless all of them are part of the \kelgroup{}. By definition, no
 candidate scoring lower than the border candidates is a member of the
 \kelgroup; thus, the term border candidates. From now on, we refer to the
 election situation characterized by parameters $\lowestscore, \toconfirm,
 \tiednr{}$ as a \emph{(input) state}. Additionally, we call a set of
 manipulators' votes a \emph{manipulation}.

 \mypar{Part 1: High-level description of the algorithm.}
 For now, we assume that there is a procedure~\pprocedure{} which runs in~\fpt{}
 time with respect to the combined parameter ``number of manipulators and number
 of different utility values.'' Procedure~\pprocedure, takes values
 \lowestscore, \toconfirm, \tiednr{} and an instance of the problem, and finds a
 manipulation which leads to a~\kelgroup{} maximizing the egalitarian utility
 under either egalitarian optimistic or egalitarian pessimistic tie-breaking
 with respect to the input state. If such a manipulation does not exist, then
 procedure \pprocedure{} returns ``no.'' The algorithm solving
 \CM{\text{\ellBloc{}}}{\Fpess^{\egal}}{\egal} and
 \CM{\text{\ellBloc{}}}{\Fopt^{\egal}}{\egal} runs \pprocedure{} for all
 possible combinations of values \lowestscore, \toconfirm{}, and \tiednr{}.
 Eventually, it chooses the best manipulation returned by \pprocedure{} or
 returns ``no'' if \pprocedure{} always returned so. Since the value of
 \lowestscore{} is at most $|V + W|$ and \tiednr{} together with \toconfirm{}
 are both upper-bounded by the number of candidates, we run \pprocedure{} at
 most $(n+\manipnumber)m^2$ times. Because the input size grows polynomially
 with respect to the growth of values \manipnumber{}, $m$, and $n$, the overall
 algorithm runs in \fpt{} time with respect to the combined parameter ``number
 of manipulators and number of different utility values.''

 \mypar{Part 2: Basics and preprocessing for the ILP.}
 To complete the proof we describe procedure \pprocedure{} used by the above
 algorithm. In short, the procedure builds and solves an ILP program that finds
 a manipulation leading to the state described by the input values. Before we
 describe the procedure in details, we start with some notation.
 Fix some values of \lowestscore{}, \tiednr{}, \toconfirm{} and some election
 $E=(C,V)$ that altogether form the input of \pprocedure{}. For each candidate
 $c \in C$, let a size-\manipnumber{} vector~$t=(u_1(c),\allowbreak u_2(c),
 \ldots ,u_\manipnumber(c))$, referred to as a \emph{type vector}, define the
 \emph{type}~of~$c$. We denote the set of all possible type vectors by
 \namedorderedsetof{\typesset}{\typevect}{\typesnr}. Observe that $\typesnr
 \leq \udiff^\manipnumber$. With each type vector $\typevect_i$, $i \in
 [\typesnr]$, we associate a set~$\typecandidates_i$ containing only
 candidates of type $\typevect_i$. We also distinguish the candidates with
 respect to their initial score compared~to~\lowestscore{}. A~candidate of type
 $\typevect_i \in \typesset$, $i \in \typesnr$, with score $\lowestscore-j$, $j
 \in \manipnumberandzero$, belongs to group \group{}. We denote all candidates
 with a score (excluding manipulative votes) higher than \lowestscore{} by
 $C^+$, whereas by $C^-$ we denote the candidates with a score (excluding
 manipulative votes) strictly lower than $\lowestscore{}-\manipnumber$. For each
 type $\typevect_i \in \typesset$ of a candidate, we define function
 $\obligatory(\typevect_i)=|C^+ \cap \typecandidates_i|$, which gives the number
 of candidates of type $\typevect{}_i$ that are obligatory part of the winning
 \kelgroup{}.

 At the beginning, procedure \pprocedure{} tests whether the input values
 \lowestscore, \tiednr, and \toconfirm{} represent a correct state. From the
 fact that there has to be at least one candidate with score \lowestscore{}, we
 get the upper bound $k-|C^+|-1$ for value $\toconfirm$. To have enough
 candidates to complete the \kelgroup{}, we need at least $k-|C^+|-\toconfirm$
 candidates with score \lowestscore{} after the manipulation which gives $\tiednr
 \geq k-|C^+|-\toconfirm$. Finally, the state is incorrect if the corresponding
 set $C^+$ contains $k$ or more candidates. If the input values are incorrect,
 then \pprocedure{} returns ``no.'' Otherwise, \pprocedure{} continues with
 building a corresponding ILP program. We give two separate ILP programs---one
 for the optimistic egalitarian tie-breaking and the other one for the
 pessimistic egalitarian tie-breaking. Both programs consist of two parts. The
 first part models all possible manipulations leading to the state described by
 values \lowestscore, \toconfirm, and \tiednr{}. The second one is responsible
 for selecting the best \kelgroup{} assuming the particular tie-breaking and
 considering all possible manipulations according to the first part.  Although
 the whole programs are different from each other, the first parts stay the
 same. Thus, we postpone distinguishing between the programs until we describe
 the second parts. For the sake of readability, we present the ILP programs step
 by step.
 
 \mypar{ILP: Common part.}
 For each group $\group$, $i \in [\typesnr]$, $j \in \manipnumberandzero$, we
 introduce variables \xpending{} and \xconfirmed{} indicating the numbers of,
 respectively, border and promoted candidates from group \group{}. Additionally,
 we introduce variables \reqapprov{} and \notspendapprov{}. The former
 represents the number of approvals used to get the obligatory numbers of border
 and promoted candidates. The latter indicates the number of approvals which are
 to be spent without changing the final \kelgroup{} (thus, in some sense a
 complement of the obligatory approvals) resulting from the manipulation (e.g.,
 approving candidates in~$C^+$, who are part of the winning \kelgroup{} anyway,
 cannot change the outcome). We begin our ILP program with ensuring that the
 values of \xpending{} and \xconfirmed{} are feasible:
 \begin{align}
  \forall{\typevect_i \in \typesset, j \in \manipnumberandzero}\colon&\xconfirmed{} +
  \xpending{} \leq |\group|, \label{eq:sum_selected}\\
  \sum_{\typevect_i \in \typesset, j \in \manipnumberandzero}\, &\xconfirmed{} =
  \toconfirm{}, \label{eq:confirmed_selected}\\ 
  \sum_{\typevect_i \in \typesset, j \in \manipnumberandzero}\, &\xpending{} =
  \tiednr{}, \label{eq:pending_selected}\\
  \forall{\typevect_i \in \typesset}\colon &\xvarname^{0+}_i + \xvarname^0_i =
  |\groupvar^0_i|,
  \label{eq:z_score_picked}\\
  \forall{\typevect_i \in \typesset}\colon
  &\xvarname^{\manipnumber+}_i = 0. \label{eq:z-r_not_confirmed}
 \end{align}
 The expressions ensure that exactly \toconfirm{} candidates are selected to be
 promoted \eqref{eq:confirmed_selected}, exactly \tiednr{} candidates are
 selected to be border ones \eqref{eq:pending_selected}, and that, for every
 group, the sum of border and promoted candidates is not greater that the
 cardinality of the group \eqref{eq:sum_selected}. The last two formulae ensure
 that candidates who have score~\lowestscore{} are either promoted or
 border candidates \eqref{eq:z_score_picked} and that candidates with initial
 score $\lowestscore{}-\manipnumber$ cannot be promoted (i.e., get a score
 higher than \lowestscore{}) \eqref{eq:z-r_not_confirmed}. Next, we add the
 constraints concerning the number of approvals we need to use to perform the
 manipulation described by all variables \xpending{} and \xconfirmed{}. We
 start with ensuring that the manipulation does not exceed the number of
 possible approvals. As mentioned earlier, we store the number of required
 approvals using variable \reqapprov.
 \begin{align}
  \reqapprov{}&=\sum_{\typevect_i \in \typesset,\ j \in \manipnumberandzero}
  \left( \xpending \cdot j + \xconfirmed \cdot (j+1)\right), \\
  \reqapprov{}&\leq \ell\manipnumber.
 \end{align} 
 \noindent Then, we model spending the \notspendapprov{} remaining votes (if
 any) to use all approvals.
 \begin{align}
  \begin{split}
   \notspendapprov&\leq \manipnumber(|C^- \cup C^+|) + \sum_{\typevect_i \in
   \typesset} \sum_{j \in [\manipnumber]} \left(|\group| - \xpending -
   \xconfirmed \right) \left( j-1 \right) \label{eq:wasted_approvals}\\
   &+\sum_{\typevect_i \in \typesset,j \in [\manipnumber]} \left( \xconfirmed
   \cdot (\manipnumber-j-1)\right),
  \end{split}\\
  \notspendapprov&+\reqapprov= \ell\manipnumber. \label{eq:last_of_ILP}
 \end{align}
 The upper bound on the number of votes one can spend without changing the
 outcome presented in equation~\eqref{eq:wasted_approvals} consists of three
 summands. The first one indicates the number of approvals which can be spent
 for candidates whose initial score was either too high or too low to make a
 difference in the outcome of the election resulting from the manipulation. The
 second summand counts the approvals we can spend for potential promoted
 and border candidates that eventually are not part of the winning \kelgroup{};
 we can give them less approvals than are needed to make them border
 candidates. The last summand represents the number of additional approvals that
 we can spend on the promoted candidates to reach the maximum of $r$ approvals
 per candidate. This completes the first part of the ILP program in which we
 modeled the possible variants of promoted and border candidates for the fixed
 state $(\lowestscore{}, \tiednr{}, \toconfirm{}$). 

 \mypar{ILP extension for optimistic egalitarian tie-breaking.}
 In the second part, we find the final \kelgroup{} by completing it with the
 border candidates according to the particular tie-breaking mechanism. Let us
 first focus on the case of the optimistic egalitarian tie-breaking. We
 introduce constraints allowing us to maximize the total egalitarian utility
 value of the final \elgroup{}; namely, for each group $\group$, $i \in
 [\typesnr]$, $j \in \manipnumberandzero$, we add a non-negative, integral variable
 \xfinmem{} indicating the number of border candidates of the given group chosen
 to be in the final \kelgroup{}.  The following constraints ensure that we select
 exactly $k-|C^+|-\toconfirm{}$ border candidates to complete the winning
 \elgroup{} and that, for each group \group{}, we do not select more candidates
 than available.
 \begin{align}
   &\sum_{\typevect_i \in \typesset,\ j \in \manipnumberandzero}\,
   \xfinmem{}=k-|C^+|-\toconfirm{}, \\
   \forall{\typevect_i \in \typesset,\ j \in \manipnumberandzero}&\colon\xfinmem{} \leq
   \xpending.
 \end{align}
 \noindent To complete the description of the ILP, we add the final expression
 defining the egalitarian utility $s$ of the final \kelitegroup{}. The goal of
 the ILP program is to maximize $s$.
 \begin{align}
   \forall{q \in [\manipnumber]}\colon \sum_{\typevect_i \in \typesset,\ j \in
    \manipnumberandzero} t_i[q]\cdot(\xconfirmed + \xfinmem) 
    + \sum_{\typevect_i \in \typesset} t_i[q] \cdot \obligatory(t_i) \geq s.
 \end{align}
 Since the goal is to maximize $s$, our program simulates the egalitarian
 optimistic tie-breaking. 
 
 \mypar{ILP extension for pessimistic egalitarian tie-breaking.}
 \newcommand{\designatedcandidate}{\ensuremath{d_i^q}}
 \newcommand{\tiebreakingwinnersnr}{\ensuremath{k-\toconfirm-|C^+|}} 
 To solve a subproblem for the case of pessimistic egalitarian
 tie-breaking, we need a different approach. We start with an additional
 notation. 
 For each type of
 candidate $\typevect_i \in \typesset$, let $\tiednr_i=\sum_{j \in
 [\manipnumber] \cup \{0\}} \xpending$ denote the number of border candidates of
 this type. For each
 type $\typevect_i \in \typesset$ and manipulator $u_q$, $q \in [\manipnumber]$,
 we introduce a new integer variable~\designatedcandidate{}.  Its value
 corresponds to the number of border candidates of type $\typevect_i$ who are
 part of the worst possible winning \kelgroup{} according to manipulator's $u_q$
 preferences; we call these candidates the \emph{designated candidates} of type
 $\typevect_i$ of manipulator~$u_q$. For each variable \designatedcandidate{}, we
 define a binary variable $\used[\designatedcandidate]$ which has value one if
 at least one candidate of type $\typevect_i$ is a designated candidate of
 manipulator~$u_q$.  Similarly, we define $\fused[\designatedcandidate]$ to
 indicate that all candidates of type $\typevect_i$ are designated by
 manipulator~$u_q$. To give a program which solves the case of pessimistic
 egalitarian tie-breaking, we copy the first part of the previous ILP program
 (expressions from \eqref{eq:sum_selected} to \eqref{eq:last_of_ILP}) and add new
 constraints.  First of all, we ensure that each manipulator designates not more
 than the number of available border candidates from each type and that every
 manipulator designates exactly \tiebreakingwinnersnr{} candidates.
 \begin{align}
  \forall{\typevect_i \in \typesset, q \in [\manipnumber] }\colon& 0 \leq
  \designatedcandidate \leq \tiednr_i,\\
  \forall{q \in [\manipnumber]}\colon&
  \sum_{\typevect_i \in \typesset}\, \designatedcandidate =
  \tiebreakingwinnersnr.
 \end{align}
 \noindent The following forces the semantics of the variables $\used$; that is,
 a variable $\used[\designatedcandidate]$, $i \in [|\typesset|]$, $q \in
 [\manipnumber]$, has value one if and only if variable \designatedcandidate{} is
 at least one.
 \begin{align}
  \forall{\typevect_i \in \typesset, q \in [\manipnumber]}\colon&
  \used[\designatedcandidate] \leq \designatedcandidate,\\
  \forall{\typevect_i \in \typesset, q \in [\manipnumber]}\colon&
  \used[\designatedcandidate]n \geq \designatedcandidate.
 \end{align}
 \noindent Similarly, for the variables $\fused$, we ensure that
 $\fused[\designatedcandidate]$, $i \in [|\typesset|]$, $q \in [\manipnumber]$,
 is one if and only if manipulator $u_q$ designates all available candidates of
 type $\typevect_i$.
 \begin{align}
  \forall{\typevect_i \in \typesset, q \in [\manipnumber]}\colon&
  \fused[\designatedcandidate] \geq 1 - (\tiednr_i - \designatedcandidate),\\
  \forall{\typevect_i \in \typesset, q \in [\manipnumber]}\colon&
  \tiednr_i - \designatedcandidate \leq n(1-\fused[\designatedcandidate]).
 \end{align}
 \noindent Since our task is to perform pessimistic tie-breaking, we have to
 ensure that the designated candidates for each manipulator are the candidates
 whom the manipulator gives the least utility. We impose it by forcing that the
 more valuable candidates (for a particular manipulator) are used only when all
 candidates of all less valuable types (for the manipulator) are used (i.e.,
 they are fully used). To achieve this we make use of the $\used{}$ and
 $\fused{}$ variables in the following constraint.
 \begin{align}
  \forall{q \in [\manipnumber] \cup \{0\}}\forall{\typevect_i, \typevect_{i'}
  \in \typesset{} \colon
   t_i[q]>t_{i'}[q]}\colon&
   \used[d_i^q] \leq \fused[d_{i'}^q].
 \end{align}
 Finally, we give the last expression where $s$ represents the pessimistic
 egalitarian \kelgroup{}'s utility which our ILP program wants to maximize:
 
 \begin{align}
  \forall{q \in [\manipnumber]}\colon \sum_{\typevect_i \in \typesset}
  (\designatedcandidate + \obligatory(\typevect_i))
  \cdot t_i[q] \geq s.
 \end{align}

 The ILP programs, for both tie-breaking variants, use at most $O(\manipnumber t)$
 variables so, according to \citet{Len83}, are in \fpt{} with respect to the
 combined parameter~$\manipnumber + \udiff$. Consequently, procedure
 \pprocedure{} is in \fpt{} with respect to the same parameter. 
 
\end{proof}

After presenting the \fpt{} result for egalitarian coalitional manipulation with
optimistic or pessimistic egalitarian tie-breaking in~\autoref{egalCMilp}, we
proceed with an analogous result for egalitarian coalitional manipulation with
one of the four remaining tie-breaking rules (that is, \{optimistic,
pessimistic\} $\times$ \{utilitiarian, can\-di\-date-wise utilitarian\})
in~\autoref{egalCMilplex}.

\begin{theorem}\label{egalCMilplex}
 Let \manipnumber{} denote the number of manipulators and \udiff{} denote the
 number of different utility values. Parameterized
 by~$\manipnumber+\udiff$, \CMlong{\text{\ellBloc{}}}{\Flex}{\egal} is
 fixed-parameter tractable.
\end{theorem}

\begin{proof}
 \newcommand{\groupvar}{\ensuremath{G}}
 \newcommand{\group}{\ensuremath{\groupvar_i^j}}
 \newcommand{\typesset}{\ensuremath{\mathcal{T}}}
 \newcommand{\typevect}{\ensuremath{t}}
 \newcommand{\typesnr}{\ensuremath{|\typesset|}}
 \newcommand{\typecandidates}{\ensuremath{T}}
 \newcommand{\manipnumberandzero}{{\ensuremath{[\manipnumber] \cup \{0\}}}}

  The general proof idea is to show an algorithm which solves problem
  \CMlong{\text{\ellBloc{}}}{\Flex}{\egal}.

  To solve \CMlong{\text{\ellBloc{}}}{\Flex}{\egal} we create an ILP program for
  all possible value combinations of the following parameters:
  \begin{itemize}
   \item the lowest final score~$z < |V \cup W|$ of any member of the
     \kelgroup{} and
   \item the candidate~$\hat{c}$ which is the least preferred member of the
     \kelgroup{} with final score~$z$ with respect to the tie-breaking
     rule~\Flex.
  \end{itemize}
 
 Having $z$~fixed, let $C^{+}$~denote the set of candidates which get at least
 $z+1$ approvals from the non-manipulative voters or which are preferred
 to~$\hat{c}$ with respect to~\calF and get exactly $z$~approvals from the
 non-manipulative voters. Assuming that the combination of parameter values is
 correct, all candidates from $C^{+} \cup \{\hat{c}\}$ must belong to the
 \kelgroup{}. We check whether $|C^+| < k$, that is, whether there is space for
 candidate~$\hat{c}$ in the \kelgroup{}. If the check fails, then we skip the
 corresponding combination of solution parameter values. Next, we ensure that
 $\hat{c}$ obtains final score exactly~$z$.  If $\hat{c}$~receives less
 than~$z-\manipnumber$ or more than~$z$ approvals from non-manipulative voters,
 then we discard this combination of solution parameter values. Otherwise,
 let~$\hat{s}:=z - \score_V(\hat{c})$ denote number of additional approvals
 candidate~$\hat{c}$ needs in order to get final score~$z$.
  
 We define the type of some candidate $c_i$ to be the size-\manipnumber{} vector
 $t_j=(u_1(c_i),\allowbreak u_2(c_i), \ldots ,u_\manipnumber(c_i))$. We denote
 by \namedorderedsetof{\typesset}{\typevect}{\typesnr} the set of all
 possible types. Observe that $\typesnr \leq \udiff^\manipnumber$. With each
 type vector $\typevect_i$, $i \in [\typesnr]$, we associate a set
 $\typecandidates_i$ containing only the candidates of type $\typevect_i$.
 Having $\hat{c}$ (and $z$) fixed, we distinguish candidates according to types
 further. For $j \in \manipnumberandzero$, all candidates with score $z-j$ that
 are preferred (resp. not preferred) to candidate $\hat{c}$ according to
 \calF{}, fall into group $\groupvar_{i}^{j+}$ (resp.\ $\groupvar_{i}^{j-}$).
 For
 each type $\typevect_i \in \typesset$ of a candidate, we define function
 $\obligatory(\typevect_i)=|C^+ \cap \typecandidates_i|$ which gives the number
 of candidates of type $\typevect{}_i$ who are obligatory part of the winning
 \kelgroup{}. We denote by $C_r$ candidates which do not fall to any of such
 groups. 
  
 We give the following ILP formulation of the problem using
 $2\manipnumber \typesnr+2$ variables. For all groups $\groupvar_i^{j+}$ and
 $\groupvar_i^{j-}$, $i \in \typesnr$, $j \in \manipnumberandzero$, we introduce
 variables $x_i^{j+}$ and $x_i^{j-}$ respectively. The variables indicate,
 respectively, the number of candidates from groups $\groupvar_i^{j+}$ and
 $\groupvar_i^{j-}$ whom we push to the winning \kelgroup{}. Also, we introduce
 two additional variables $s$ and $u$. The former one represents the minimal
 value of the total utility achieved by manipulators. The latter one indicates
 the number of votes which were spent without changing the outcome. To shorten
 the ILP we define
 $$
 \mful_{z}^{\hat{c}}:=\sum_{\typevect_i \in \typesset ,j \in
 [\manipnumber]} x_i^{j+} \cdot j + \sum_{\typevect_i \in \typesset, j \in
 [\manipnumber-1] \cup \{0\}} x_i^{j-} \cdot (j+1).
 $$
 Intuitively, $\mful_{z}^{\hat{c}}$ is the number of approvals used to make
 potential winners the winners. Also, we define
 $$
 \fbid_{z}^{\hat{c}}:=\sum_{\typevect_i \in \typesset, j \in
 \manipnumberandzero} \left [(\manipnumber-j+1)(|\groupvar_i^{j+}|-x_i^{j+})+
 (\manipnumber-j)(|\groupvar_i^{j-}|-x_i^{j-}) \right].
 $$
 $\fbid_{z}^{\hat{c}}$ represents the number of approvals which cannot be used
 if one wants to avoid pushing candidates outside of the solution (given by
 values of the variables $x$) to the winning \kelgroup{}; for example, if some
 candidate $c$ needs $j$ approvals to be part of the winning committee, then we
 subtract $r-j+1$ approvals from the whole pool of $r$ approvals for this
 candidate because we can use only $j-1$ approvals not to push~$c$ into the
 \kelgroup{}. We define the following constraints to construct our program the
 goal of which is to maximize
 $s$:
 \begin{align}
  \forall{\typevect_i \in \typesset, j \in \manipnumberandzero, \bullet \in
  \{+,-\}}\colon& x_i^{j\bullet} \le |\groupvar_i^{j\bullet}|
  \label{for:upper-cand},\\
  \forall{\typevect_i \in \typesset}\colon&x_i^{z-} = 0 \label{for:no-unn},\\
  \forall{\typevect_i \in \typesset}\colon&x_i^{0+} = |\groupvar_i^{0+}|
  \label{for:best-border-in},\\
  &\mful_{z}^{\hat{c}} \le \manipnumber \cdot \ell - \hat{s}
  \label{for:upper-mful},\\
  &u \le (|C|-1)\manipnumber{} - \mful_{z}^{\hat{c}} - \fbid_{z}^{\hat{c}},
  \label{for:upper-wasted}\\
  &u + \mful_{z}^{\hat{c}} = \manipnumber \cdot \ell - \hat{s}
  \label{for:upper-all},\\
  &\sum_{\typevect_i \in \typesset, j \in \manipnumberandzero, \bullet \in \{+,-\}}
  x_i^{j\bullet} = k - |C^+| - 1 \label{for:kelgroup},\\
  \forall{q \in [\manipnumber]} \colon &\sum_{\typevect_i \in \typesset, j \in
   \manipnumberandzero,
   \bullet \in \{+,-\}} x_i^{j\bullet} \cdot t_i[q] + \nonumber\\
   &\sum_{\typevect_i \in \typesset} t_i[q] \cdot \obligatory(t_i) \geq s.
  \label{for:utils}
 \end{align}
	 
 Constraint~\eqref{for:upper-cand} ensures that the candidates picked into a
 solution are available and can be part of the solution. Observe that candidates
 in $\groupvar_i^{0+}$ have to be part of the solution and candidates in
 $\groupvar_i^{z-}$ cannot be part of the solution. These two facts are ensured
 by Constraints~\eqref{for:no-unn} and \eqref{for:best-border-in}.
 Constraint~\eqref{for:upper-mful} forbids spending more votes than possible to
 push some candidates to the \kelgroup{}. The same role for ``wasted'' approvals
 plays Constraint~\eqref{for:upper-wasted}. The upper bound of wasted approvals
 is counted in the following way: From the number of all ``places'' of putting
 approvals (we subtract one from the number of candidates because we cannot put
 any approvals except for~$\hat{s}$ to candidate~$\hat{c}$), we first subtract
 the approvals already given to candidates in the \kelgroup{} (i.e.,\
 $\mful_{z}^{\hat{c}}$). Next, we subtract all ``places''  of approvals that
 will cause the unchosen potential candidates to be chosen (i.e.,\
 $\fbid_{z}^{\hat{c}}$).
 Constraint~\eqref{for:upper-all} ensures that, altogether, we spend exactly as
 many approvals as required, and Constraint~\eqref{for:kelgroup} holds only when
 a proper number of candidates are pushed to be part of \kelgroup{}. The last
 equation forces maximization of the egalitarian utility of the winning
 \kelgroup{} when $s$ is maximized.
     
 Using our technique we can obtain a solution by making $O(nm)$ ILPs with at
 most $2\manipnumber \udiff^{\manipnumber}+2$ variables. According to Lenstra's
 famous result~\citep{Len83}, the constructed ILPs yield fixed-parameter
 tractability with respect to the combined parameter $\manipnumber + \udiff$.
\end{proof}

\section{Conclusion}
\label{sec:conclusion}
\begin{table*}
 \caption{Computational complexity of tie-breaking and coalitional manipulation.
  Our results for \ellBloc{} hold for any~$\ell\geq 1$, and thus cover \SNTV.
  The parameters are the size~$k$ of the \elgroup{}, the
  number~\manipnumber{} of manipulators, and the number~\udiff{} of different
  utility values. Furthermore, $m$~is the number of candidates and~$n$ is the
  number of voters. The result marked with $\dagger{}$ holds for all possible
  combinations of the respective evaluation and behavior variants. The results
  marked with~$\diamond$ hold also for $\calF=\Flex$.}
  \centering
 \begin{tabular}{l c l}
  \multicolumn{3}{l}{\TIElong{\eval}{\bhav}, easy cases:}\\
  \toprule
  settings (evaluation, behavior) & complexity & reference\\
  \midrule
  utilitarian or cand.\-wise egalitarian,   & \multirow{2}{*}{$O(m \cdot
  (\manipnumber + \log m))$} & \multirow{2}{*}{Cor.~\ref{cor:TBeasy} $\dagger$} \\
  optimistic or pessimistic &                                \\[1ex]
  egalitarian, pessimistic     & $O(\manipnumber \cdot m \log m)$ &
  Thm.~\ref{thm:egalTB}\\
 \end{tabular}

 \vspace{3ex}

 \begin{tabular}{l c l}
  \multicolumn{3}{l}{\TIElong{\egal}{\opt} (egalitarian, optimistic):}\\
  \toprule
  parameters, restrictions               & complexity & reference \\
  \midrule
  general                  & $\np$-hard            & Thm.~\ref{thm:egalTB} \\
  $k$, 0/1 utilities and $q=1$ & $\wtwo$-hard              & Thm.~\ref{thm:egalTB} \\
  $\manipnumber+k$         & $\wone$-hard              & Thm.~\ref{thm:egalTBcomb} \\
  $\manipnumber+\udiff$     & $\fpt$                & Thm.~\ref{thm:egalTBilp} \\
 \end{tabular}

 \vspace{3ex}

  \begin{tabular}{l c l}
  \multicolumn{3}{l}{\CMlong{\ellBloc}{\calF}{\eval}}\\
  \multicolumn{3}{l}{utilitarian/cand.\-wise egalitarian, optimistic/pessimistic:}\\
  \toprule
  restrictions                  & complexity & reference \\
  \midrule
  general                      & $O(k^2m^2(n+r))$ &
  Thm.~\ref{thm:genCMinP} $\diamond$\\
  consistent manipulators      & $O(m(m + r +n))$ & Prop.~\ref{prop:utilConsistentPoly} $\diamond$\\
  $\ell=k$                     & $O(m(m + r +n))$ & Cor.~\ref{cor:easyCM} $\diamond$\\
 \end{tabular}

 \vspace{3ex}

 \begin{tabular}{l c l}
  \multicolumn{3}{l}{\CMlong{\ellBloc}{\calF}{\eval}}\\
  \multicolumn{3}{l}{egalitarian, optimistic/pessimistic:}\\
  \toprule
  parameters, restrictions & complexity & reference \\
  \midrule
  general         & $\np$-hard                   & Cor.~\ref{cor:egalCMhard} $\diamond$\\
  $k$, 0/1 utilities and $q=1$ & $\wtwo$-hard                     & Cor.~\ref{cor:egalCMhard} $\diamond$\\
  $\manipnumber+k$   & $\wone$-hard                     & Cor.~\ref{cor:egalCMhard} $\diamond$\\
  $\manipnumber+\udiff$ & $\fpt$                       &
  Thm.~\ref{egalCMilp} and Thm.~\ref{egalCMilplex} $\diamond$\\
 \end{tabular}

\label{tab:results}
\end{table*}
We developed a new model for and started the first systematic study of
coalitional manipulation for multiwinner elections.
Our analysis revealed that multiwinner coalitional manipulation requires models
which are significantly more complex than those for single-winner coalitional manipulation
or multiwinner non-coalitional manipulation.
Our model assumes a given, fixed coalition of manipulators can compensate
their (potential) utility loss after a manipulation in some way. Thus, in
particular, our model does not account for questions like whether a particular
coalition can be stable, how coalitions are forming, or what to do to avoid a
coalition's split. Yet, we think that there are two important reasons why the
model is, in fact, useful. First, there are situation where manipulators will not leave
the coalition despite of being worse off after some manipulation. Such a
situation might naturally occur if manipulators share a common goal and they
agree either on losing their individual utility for the good of the whole group
or on compensating their utility losses internally among themselves.
Second, assessing the quality of possible manipulations for a given coalition is
essential to answer more general questions about coalitions (e.g.,\ what is the
most profitable coalition) that lead to a new research direction (see the last
paragraph of this section for a broader discussion on this direction).

In our work, on the one hand, we generalized tractability results for
coalitional manipulation of $\ell$-Approval by \citet{CSL07} and \citet{Lin11}
and for non-coalitional manipulation of \Bloc{} by \citet{MPRZ08} and
\citet{OZE13} to tractability of coalitional manipulation of \ellBloc{} in case
of utilitarian or candidate-wise egalitarian evaluation of \elgroup{}s. On the
other hand, we showed that coalitional manipulation becomes 
intractable in case of egalitarian evaluation of \elgroup{}s.

Let us discuss a few findings in more detail (Table~\ref{tab:results} surveys
all our results). We studied lexicographic, optimistic, and pessimistic
tie-breaking and showed that, with the exception of egalitarian group
evaluation, winner groups can be determined very efficiently. The intractability
(NP-hardness, parameterized hardness in form of \wone{}- and \wtwo{}-hardness)
for the egalitarian case, however, turns out to hold even for quite restricted
scenarios. We also demonstrated that numerous tie-breaking rules can be
``simulated'' by (carefully chosen) lexicographic tie-breaking, again except for
the egalitarian case. Interestingly, the hardness of egalitarian
tie-breaking holds only for the optimistic case while for the pessimistic case
it~is~efficiently solvable. Hardness for the egalitarian optimistic scenario,
however, translates into hardness results for coalitional manipulation
\emph{regardless} of the specific tie-breaking rule. On~the~contrary,
coalitional manipulation becomes tractable for the other two evaluation
strategies---``candidate-wise'' egalitarian and utilitarian. Additionally, for
few candidates and few different utility values the voters assign to the
candidates, manipulation becomes tractable also for the egalitarian optimistic
scenario.

In our study, we entirely focused on shortlisting as one of the simplest tasks
for multiwinner elections to analyze our evaluation functions. It is interesting
and non-trivial to develop models for multiwinner rules that aim for
proportional representation or diversity. For shortlisting, extending our
studies to non-approval-like scoring-based voting correspondences would be
a~natural next step. In this context, already seeing what happens if one extends
the set of individual scores from being only 0 or~1 to more (but few) numbers is
of interest. Moreover, we focused on deterministic tie-breaking mechanisms,
ignoring randomized tie-breaking---another issue for future research.

An analysis of the manipulators' behavior, briefly mentioned at the beginning of
this section, directing towards game theory seems promising as well. (Even more
so since we identified polynomial-time algorithms for a few variants of
coalitional manipulation.) One very interesting question about coalitions is,
for example, whether a particular coalition is stable. Intuitively, the utility
for~every voter that is a part of~the manipulating coalition should not be below
the~utility the voter receives when voting sincerely. This is of course only a
necessary condition to ensure the stability of a coalition. A more sophisticated
analysis of stability needs to consider game-theoretic aspects such as Nash or
core stability~\citep{NRTV07}.

\subsection*{Acknowledgments}
We thank the anonymous \emph{IJCAI '17} reviewers for their constructive
and valuable feedback.
Robert Bredereck was from mid-September 2016 to mid-September 2017 on
postdoctoral leave at the University of Oxford, supported by the DFG fellowship
BR 5207/2. Andrzej Kaczmarczyk was supported by the DFG project AFFA (BR 5207/1
and NI 369/15).

\bibliographystyle{plainnat}
\bibliography{bibliography}

\end{document}

%% file: commons.tex
\newcommand{\orderedsetof}[2]{\ensuremath{\{{#1}_1, {#1}_2, \ldots, {#1}_{#2}\}}}

\newcommand{\namedorderedsetof}[3]{\ensuremath{{#1}=\orderedsetof{#2}{#3}}}